%% file: TreewidthPLS.tex
\documentclass[11pt, oneside]{article}   	

\usepackage{authblk}

\usepackage[a4paper,margin=2.5cm]{geometry}
\usepackage{graphicx}				
\usepackage[utf8]{inputenc}
\usepackage[T1]{fontenc}

\usepackage{amsmath,amsfonts,amssymb,amsthm}
\usepackage{todonotes}
\usepackage{tikz} 
\theoremstyle{plain}
  \newtheorem{proposition}{Proposition}
  \newtheorem{theorem}{Theorem}
  \newtheorem{lemma}{Lemma}
  \newtheorem{corollary}{Corollary}
\theoremstyle{definition}
  \newtheorem{definition}{Definition}
  
\newcommand{\m}[1]{\mathcal{#1}}
\newcommand{\tw}{\mathsf{tw}}

\newcommand{\cO}{\mathcal{O}}
\newcommand{\cT}{\mathcal{T}}
\newcommand{\cB}{\mathcal{B}}
\newcommand{\cF}{\mathcal{F}}
\newcommand{\cP}{\mathcal{P}}
\newcommand{\cC}{\mathcal{C}}
\newcommand{\cE}{\mathcal{E}}
\newcommand{\ID}{\mathsf{ID}}

\title{A Meta-Theorem for Distributed Certification\thanks{This work was partially done during the visit of the second and third authors to IRIF at Universit\'e de Paris, and LIFO at Universit\'e d'Orl\'eans, partially supported by ANR project DUCAT and Fondecyt 1170021.}}
\author[1]{Pierre Fraigniaud\thanks{Additional support for ANR projects QuData and DUCAT.}}
\author[2]{Pedro Montealegre\thanks{This work was supported by Centro de Modelamiento Matem\'atico (CMM), ACE210010 and FB210005, BASAL funds for centers of excellence from ANID-Chile, and FONDECYT 11190482}}
\author[3]{Ivan Rapaport}
 \author[4]{Ioan Todinca}

\affil[1]{\small IRIF, Universit\'e de Paris and CNRS, France. \texttt{pierre.fraigniaud@irif.fr}}    
         
\affil[2]{Facultad de Ingenier\'{\i}a y Ciencias, Universidad Adolfo Iba\~nez, Santiago, Chile. \texttt{p.montealegre@uai.cl}} 
               
\affil[3]{DIM-CMM (UMI 2807 CNRS), Universidad de Chile, Chile. \texttt{rapaport@dim.uchile.cl}}        
    	 
\affil[4]{LIFO, Universit\'e d'Orl\'eans and INSA Centre-Val de Loire, France. \texttt{ioan.todinca@univ-orleans.fr}}

\date{}							

\begin{document}
\maketitle

\begin{abstract}
Distributed certification, whether it be \emph{proof-labeling schemes}, \emph{locally checkable proofs}, etc.,  deals with the issue of certifying the legality of a distributed system with respect to a given boolean predicate. A certificate is assigned to each process in the system by a non-trustable oracle, and the processes are in charge of verifying these certificates, so that two properties are satisfied: \emph{completeness}, i.e., for every  legal instance, there is a certificate assignment leading all processes to accept, and \emph{soundness}, i.e., for every illegal instance, and for every certificate assignment, at least one process rejects. The  verification of the certificates must be fast, and the certificates themselves  must be  small.  A large quantity of results have been produced in this framework, each aiming at designing a distributed certification mechanism for specific boolean predicates. This paper presents a ``meta-theorem'', applying to many boolean predicates at once. Specifically,  we prove that, for every boolean predicate on graphs definable in the monadic second-order (MSO) logic of graphs, there exists a distributed certification mechanism using certificates on $O(\log^2n)$ bits in $n$-node graphs of bounded treewidth, with a verification protocol involving a single round of communication between neighbors. 

\end{abstract}


%
%

\input{intro.tex}

%
%
\input{prelim.tex}

%
%
\input{prep.tex}

\input{twApprox.tex}

%
%
\input{regular.tex}

%
%
\input{conclusion.tex}

\bibliographystyle{plain}
\bibliography{Biblio}

%
%
\appendix
\input{appendixOptim.tex}

\end{document}

%% file: intro.tex

\section{Introduction}

\subsection{Context}

Distributed certification  is a concept that serves many purposes in distributed computing. One is fault tolerance. Indeed, the ability to certify the legality of a  system-state with respect to some boolean predicate in a distributed manner guarantees that at least one process can launch a recovery procedure in case the system enters into an illegal state. Another application of distributed certification is safety. Indeed,  distributed certification is a mechanism that guarantees that distributed algorithms dedicated to  systems satisfying some specific property (e.g., algorithms dedicated to planar networks) can safely be used because, in case the system does not satisfy this property, at least one process can raise an alarm, and stop the computation.  

Different certification mechanisms have been studied (cf. the related work section), all sharing the same principle. Distributed certification protocols involve a centralized \emph{prover}, and a distributed \emph{verifier}. The prover has complete knowledge of the system. It is computationally unbounded but not trustable. Given a boolean predicate~$\m{P}$ on system states, the prover assigns \emph{certificate} to the processes, whose aim is to convince the processes that the system satisfies~$\m{P}$.  The verifier is a distributed algorithm that runs at every process in the system, and is bounded to return a verdict (\emph{accept} or \emph{reject}) at each process after a limited communication among the processes. For instance, in a network, every processing node is bounded to communicate only once with its neighbors in the network before emitting its verdict. 

To be correct, a distributed certification protocol for a boolean predicate~$\m{P}$ on system states must satisfy two properties. 
(1)~Completeness: If the system satisfies~$\m{P}$, then there must exist a certificate assignment by the prover to the processes such that the verifier  accepts at all processes.  
(2)~Soundness: If the system does not satisfy~$\m{P}$, then, for every certificate assignment by the prover to the processes, it must be the case that the verifier rejects in at least one process. 
Network bipartiteness yields a simple example of distributed certification, using 1-bit certificates. For every bipartite network, every processing node in the network can be given a certificate~0 or~1, so that every processing node has a certificate different from the certificates assigned to its neighbors. The processing nodes can check these certificates in a single round of communication, where every processing node merely checks that the certificate of each of its neighbors is different from its own certificate. Completeness is satisfied by construction. Soundness is also satisfied. Indeed, if the network is not bipartite, then it is not 2-colorable. As a consequence, for every certificate assignment with certificates in~$\{0,1\}$, there are at least two neighboring processing nodes that receive the same certificate. These two processes will reject. 

The main criterion measuring the quality of distributed certification is the \emph{size} of the certificates. Indeed, the verification of~$\m{P}$ is typically performed frequently, for regularly checking that the system does satisfy~$\m{P}$, with the aim of reacting quickly if the system stops satisfying~$\m{P}$. As a consequence, there are frequent exchanges of certificates between the processes. Using small certificates limits the communication overhead caused by these exchanges. 

\subsection{Objective}

A large collection of results related to distributed certification have been derived over the last twenty years (see Related Work), each result concerning a specific predicate. This paper is inspired by what has been achieved in the context of sequential computing where, instead of focusing on the design of an efficient algorithm for one specific problem, and then for another one, and so on and so forth, efforts have been made for deriving ``meta-theorems'', that is, results applying directly to large classes of problems. One prominent example is Courcelle's theorem~\cite{Courcelle90} stating that every graph property definable in the monadic second-order (MSO) logic of graphs can be decided in linear time on graphs of bounded treewidth\footnote{Treewidth can be viewed as a measure capturing ``how close'' a graph is from a tree; roughly, a graph of treewidth~$k$ can be decomposed by a sequence of cuts, each involving a separator of size $O(k)$.}. That is, even NP-hard problems such as vertex-coloring, minimum dominating set, minimum vertex cover, etc., have linear-time algorithms in the vast class of graphs with bounded treewidth. Each algorithm depends on the problem, but Courcelle's theorem essentially says that \emph{every} problem expressible in the MSO logic has a linear-time algorithm in the class of graphs with bounded treewidth. 

The objective of this paper is to address the existence of similar meta-theorems in the context of distributed certification applied to distributed computing in networks. Concretely, 
the question we address here is the following: is there a (large) class of boolean predicates on graphs  for which one can guarantee the existence of a distributed certification mechanism with small certificates, say poly-logarithmic in the number of vertices of the graphs, for graphs taken from a (large) class of graphs? 

\subsection{Our Results}

We present an analog of the aforementioned Courcelle's theorem in the context of distributed certification. Specifically, for every integer $k\geq 1$ and every MSO property~$\varphi$ on graphs, we consider the following set: 
\[
\m{P}_{k,\varphi}=\{\mbox{graph}\; G: (\tw(G)\leq k) \wedge (G\models\varphi)\}, 
\]
where $\tw(G)$ is the treewidth of~$G$. We provide a distributed certification mechanism for~$\m{P}_{k,\varphi}$ using certificates of poly-logarithmic size, as a function of the number~$n$ of vertices in the graphs. Specifically, given any network modeled as a connected simple graph $G=(V,E)$, with a process running at each vertex $v\in V$, our certification mechanism  satisfies that $G\in\m{P}_{k,\varphi}$ if and only if there is a certificate assignment to the vertices such that all vertices accept.  The main result of the paper is the following. 

\begin{theorem}[Informal]\label{theo:main-informal}
For every $k\geq 1$ and every MSO property~$\varphi$ on graphs, there exists a distributed certification protocol for~$\m{P}_{k,\varphi}$ using certificates on $O(\log^2n)$ bits. 
\end{theorem}

In fact, our theorem can be extended to properties including certifying solutions to maximization or minimization problems whose admissible solutions are defined by MSO properties. In the statement of Theorem~\ref{theo:main-informal}, the big-O notation hides constants that depend only on~$k$ and~$\varphi$. The theorem has many corollaries, as the universe of MSO properties is large.  This includes predicates such as non 3-colorability, which is known to require certificates of quadratic size in arbitrary graphs~\cite{GoosS16}, and diameter at most~$D$, for a fixed constant~$D$, which is known to require certificates of linear size in arbitrary graphs~\cite{Censor-HillelPP20}. 

\begin{corollary}\label{cor:non-3-coloring-informal}
For every $c\geq 1$, there exists a distributed certification protocol for certifying \emph{non $c$-colorability} in the family of graphs with bounded treewidth, using certificates on $O(\log^2n)$ bits. 

For every $D\geq 1$, there exists a distributed certification protocol for certifying \emph{diameter at most~$D$} in the family of graphs with bounded treewidth, using certificates on $O(\log^2n)$ bits. 
\end{corollary}

Also, many natural graph families have bounded treewidth, as illustrated by the family of graphs excluding a planar graph as a minor, and thus we get the following corollary of Theorem~\ref{theo:main-informal}. 

\begin{corollary}\label{cor:planar-minor-free-informal}
For every planar graph~$H$, and every MSO property~$\varphi$ on graphs, there exists a distributed certification protocol certifying $\varphi$ in the family of $H$-minor-free graphs, using certificates on $O(\log^2n)$ bits. 
\end{corollary}

Again, the big-O notation in the above statement hides constants that depend only on~$H$ and~$\varphi$. Note that, as every 4-node graph is planar, Corollary~\ref{cor:planar-minor-free-informal} extends the recent results in~\cite{BoFePi21-Hfree}, which applies to the families of graphs excluding a given 4-node graph~$H$ as a minor. 

Interestingly, $\tw(G)\leq k$, and $H$-minor-freeness are themselves MSO properties for fixed~$k$ and~$H$. It follows that Theorem~\ref{theo:main-informal} provides us with a distributed certification mechanism for treewidth and fixed-minor-freeness. 

 \begin{corollary}\label{cor:treewidth-only-informal}
Let $k\geq 0$, and let $H$ be a planar graph.
There exist distributed certification protocols for certifying the class of graphs with treewidth at most~$k$, and certifying the class of  $H$-minor-free graphs, both using certificates on $O(\log^2n)$ bits. 
\end{corollary}

\paragraph{Our Techniques.} 

For establishing Theorem~\ref{theo:main-informal} we proceed in two steps. First, we provide a  protocol for certifying 3-approximation of  treewidth. 
Such a protocol satisfies the following: for any given $k\geq 1$, the protocol for $k$ is such that, for every graph~$G$, 
\[
\left\{\begin{array}{lcl}
\tw(G)\leq k & \Rightarrow & \mbox{there exists a certificate assignment s.t. all vertices accept;}\\ 
\tw(G) > 3k+2 & \Rightarrow & \mbox{for every certificate assignment, at least one vertex rejects.}\\ 
\end{array}\right.
\] 

\begin{lemma}[Informal]\label{lem:technique1-informal}
For every $k \geq 1$ there exists a distributed protocol certifying a 3-approxima\-tion of the treewidth using certificates on $O(k^2\log^2n)$ bits. 
\end{lemma}

The proof of this lemma relies on a particular choice of a tree-decomposition, that we prove locally certifiable by ``transferring'' certificates between nodes that are far away from each other, which is typically the case of vertices in a same bag of the decomposition, without creating congestion. 

Next, for any MSO property~$\varphi$ and integer $k$, we design a protocol which certifies $\m{P}_{k,\varphi}$ on input graph $G$. The protocol exploits the tree decomposition in the proof  of Lemma~\ref{lem:technique1-informal}, for certifying a correct execution of a sequential dynamic programming algorithm for~$\varphi$ over this decomposition. Concretely, we design a distributed certification for a correct execution of a sequential dynamic programming algorithm a la Courcelle, using in fact the sequential MSO certification due to Borie, Parker and Tovey~\cite{BoPaTo92}. 

\begin{lemma}[Informal]\label{lem:technique2-informal}
For every $k\geq 1$ and every MSO property~$\varphi$ on graphs,
assuming given the certification protocol for 3-approximation~$k$ of treewidth from Lemma~\ref{lem:technique1-informal},
there exists a distributed certification protocol for~$\m{P}_{k,\varphi}$ using additional certificates on $O(\log^2n)$ bits. 
\end{lemma}

\subsection{Related Work}

The ability to detect illegal configurations of a distributed system was originally motivated by the design of fault-tolerant algorithms, especially self-stabilizing algorithms~\cite{AfekKY97,AwerbuchPV91,ItkisL94}. The notion of distributed certification as used in this paper originated from the seminal paper~\cite{KormanKP10} defining \emph{proof-labeling schemes} (PLS). We actually use a slight variant of PLS called \emph{locally checkable proofs} (LCP)~\cite{GoosS16}, which enables exchanging not only the certificates between the processing nodes, but also local states, including their IDs. Another related notion is \emph{non-deterministic local decision} (NLD)~\cite{FraigniaudKP13} in which the certificates must not depend on the IDs given to the processing nodes. Distributed certification has been extended to various directions, including randomized PLS~\cite{FraigniaudPP19}, approximate PLS~\cite{Censor-HillelPP20,EmekG20}, local hierarchies~\cite{BalliuDFO18,FeuilloleyFH21}, interactive proofs~\cite{KolOS18,NaorPY20}, and even, recently, zero-knowledge distributed certification~\cite{BKO22}. All the aforementioned papers contain a vast collection of certification results for various graph problems. In these papers, each certification protocol is specific of the problem at hand. To our knowledge, the only ``meta-theorem'' in the context of distributed certification is the recent paper~\cite{BoFePi21-MSO}, which shows that every MSO formula can be locally certified on graphs with bounded \emph{treedepth} using certificates on $O(\log n)$ bits. We show that the same result holds for the larger class of graphs with bounded \emph{treewidth}, to the cost of slightly larger certificates, on $O(\log^2 n)$ bits. We are therefore partially answering the questions raised in~\cite{BoFePi21-MSO}, asking whether it is ``possible to certify any MSO formula on bounded treewidth graphs'', and ``to certify that the graph itself has treewidth at most~$k$'', using small certificates. 

In framework of sequential algorithms, there is a large literature on ``meta-theorems'' proving that large families of combinatorial properties (typically expressed using some form of logic formulae) can be efficiently decided on particular graph classes. In addition to Courcelle's (meta) theorem~\cite{Courcelle90} on MSO properties on graphs with bounded treewidth, it is worth mentioning the recent results establishing that properties expressible in \emph{first-order logic} can be verified in polynomial time on graphs of bounded \emph{twinwidth}~\cite{BKTW20}, as well as on \emph{nowhere-dense} graphs~\cite{GKS17}. Both graph classes include planar graphs, and thus include graphs with arbitrarily large treewidth. Our work is participating to the general objective of  extending these results to the framework of distributed computing.

%% file: prelim.tex

\section{Preliminaries}

\subsection{Distributed Certification} 

We consider networks modeled as connected simple graphs. Every vertex is a processing element, and the vertices exchange messages along the edges of the graph. We systematically denote by $n$ the number of vertices in the considered graph. The vertices of a network/graph $G=(V,E)$ are given distinct identifiers (IDs), and we denote by $\ID(v)$ the identifier of vertex~$v\in V$. These identifiers are not necessarily between~1 and~$n$, but we adopt the standard assumption stating that IDs can be stored on $O(\log n)$ bits. 

We consider boolean predicates on labeled graphs, i.e., graphs for which every vertex~$v$ is given a label $\ell(v)\in\{0,1\}^*$. These labels may represent a way to mark vertices (e.g., those in a dominating set), a color (e.g., in graph coloring), or any value depending on the graph property at hand. Given a boolean predicate $\m{P}$ on labeled graphs, a \emph{locally checkable proof}~\cite{GoosS16} 
for  $\m{P}$ is a prover-verifier pair. The prover is a non-trustable oracle with unbounded computing power. Given any labeled graph $(G,\ell)$, the prover assigns a certificate $c(v)\in\{0,1\}^*$ to every vertex $v\in V$. The verifier is a 1-round distributed algorithm running at all vertices of the graph. Given a labeled graph $(G,\ell)$ with a certificate assigned at every vertex, the vertices exchange their identifiers, labels, and certificates, between neighbors, and compute an output, accept or reject. To be correct, the pair prover-verifier must satisfy two conditions: 
\begin{description}
\item[Completeness:] If $(G,\ell)\models \m{P}$, then, for every ID-assignment to the vertices, there must exist a certificate assignment by the prover to the vertices such that the verifier accepts at all vertices.  
\item[Soundness:] If $(G,\ell)\not\models \m{P}$, then, for every ID-assignment to the vertices, and for every certificate assignment by the prover to the vertices, it must be the case that the verifier rejects in at least one vertex. 
\end{description}

\subsection{Tree Decompositions and Terminal Recursive Graphs} 

Let us recall the classical definition of treewidth and tree decompositions, due to Robertson and Seymour~\cite{RoSe84}.

\begin{definition}\label{de:treedec} 
A \emph{tree decomposition} of a graph $G = (V,E)$ is a pair $(T, B)$  where ${T=(I,F)}$ is a tree, 
and $B=\{B_i, i \in I\}$ is a collection of subsets of vertices of $G$, called \emph{bags}, such that the following conditions hold:
\begin{itemize}
\item For every $v \in V$, there exists $i\in I$ such that $v \in B_i$;
\item For every $e = \{u,v\} \in E$ there is $i\in I$ such that $\{u, v\} \subseteq B_i$;
\item For every $v \in V$, the set $\{i \in I : v \in B_i\}$ forms a connected subgraph of $T$. 
\end{itemize}
The \emph{width} of a tree decomposition is the maximum size of a bag, minus one. The \emph{tree-width} of a graph $G$, denoted by~$\tw(G)$, is the smallest width of a tree decomposition of~$G$.
\end{definition}

To facilitate the distinction between the original graph $G = (V,E)$ and the decomposition tree $T=(I,F)$, we will speak of the \emph{nodes} $i \in I$ of $T$ and of the \emph{vertices} $v \in V$ of $G$. 

We consider tree decompositions as rooted, i.e., we fix some node $r \in I$ as the root of $T = (I, F)$. For a node $i \in I\setminus\{r\}$, we denote by $p(i)$ its parent in $T$, and set $p(r)=\bot$. For $i\in I$, we denote by $T_i$ the subtree of $T$ rooted in $i$, and by $V_i$ the subset of vertices of~$G$ in the bags of $T_i$, i.e., $V_i=\cup_{j\in V(T_i)}B_j$. Also, for $i \in I\setminus\{r\}$, we define $F_i = B_i \setminus B_{p(i)}$.  For the root~$r$, we set $B_{p(r)}  = \varnothing$ and $F_r = B_r$.  Given a rooted tree $T = (I, F)$, and two nodes of $i,j \in I$, we denote by $j \preceq i$ the property that $j$ is a descendant of $i$ in~$T$. 

 Graphs of bounded treewidth can also be defined recursively, based on a graph grammar. Let $w$ be a 
 positive integer. 
 A \emph{$w$-terminal graph} is a graph  $(V,E)$ together with a \emph{totally ordered} set $W \subseteq V$ of at most $w$ distinguished vertices. Vertices of $W$ are called the \emph{terminals} of the graph, and we denote by $\tau(G)$ the number of its terminals. Since $W$ is totally ordered, we can speak of the $r$th terminal, for 
$1\leq r \leq w$. 
 Since in our case vertices are given distinct identifiers, one can view $W$ as ordered w.r.t. these identifiers.

The class of \emph{$w$-terminal recursive graphs} is defined starting from \emph{$w$-terminal base graphs} through a sequence of \emph{composition operations}. A \emph{$w$-terminal base graph} is a $w$-terminal graph of the form $(V,W,E)$ with $W=V$. A \emph{composition operation} $f$ acts on one or two $w$-terminal graphs producing a new $w$-terminal graph as follows. 

When $f$ is of arity 2,  graph $G = f(G_1, G_2)$ is obtained by firstly making disjoint copies of the two graphs $G_1$ and $G_2$, then ``glueing'' together some terminals of $G_1$ and $G_2$. The glueing performed by  $f$ is represented by a matrix $m(f)$ having $\tau(G) \leq w$ rows and two columns, with integer values between $0$ and $\tau(G)$. At row $r$ of the matrix, $m_{rc}(f)$ indicates which terminal of each $G_c, c \in \{1,2\}$ is identified to terminal number $r$  of graph $G$. If $m_{rc}(f) = 0$, then no terminal of $G_c$ is identified to terminal $r$ of $G$ (in particular, if  $m_{r1}(f) = m_{r2}(f) = 0$ it means that terminal $r$ of $G$ is a new vertex, but this situation will not occur in our constructions). Moreover, a terminal of $G_c$ is identified to at most one terminal of $G$, i.e., each non-zero value in $1,\dots, \tau(G_c)$  appears at most once in column $c$ of $m(f)$. For an illustration, see, e.g., Figure~\ref{fig:twgram}.

When $f$ is of arity 1, the corresponding matrix $m(f)$ has a unique column. Graph $G = f(G_1)$ is obtained as before, by identifying terminal $m_{i1}$ of $G_1$ to terminal $r$ of $G$. Note that in this case $G$ and $G_1$ have exactly the same vertex and edge sets, and the terminals of $G$ form a subset of the terminals of $G_1$. 

We point out that the number of possible different matrices and hence of different operations is bounded by a function on $w$. 

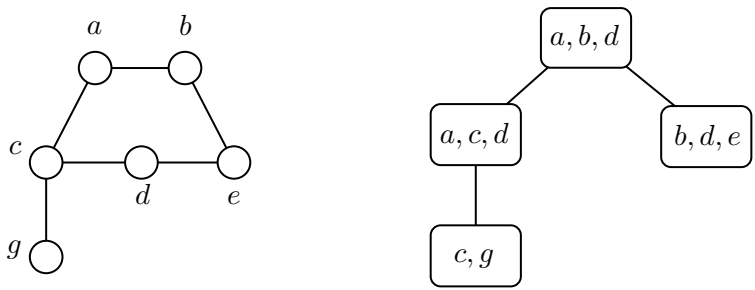
\begin{figure}[htbp]
\begin{center}
\input{Figures/ExTrDec2.tex}
\caption{Graph $G$ and a tree decomposition.}
\label{fig:ExTrDec}
\end{center}
\end{figure}

\begin{proposition}[Theorem 40 in~\cite{Bodlaender98arb}]\label{pr:gramtw}
Graph $H=(V,W,E)$ is $(w+1)$-terminal recursive if and only if there exists a tree decomposition of $G=(V,E)$, of width at most $w$, having $W$ as root bag. Hence the grammar of $(w+1)$-terminal recursive graphs constructs exactly the graphs of treewidth at most $w$. 
\end{proposition}

Let us sketch briefly here how a tree decomposition of $G=(V,E)$ of width $w$ can be transformed into a $(w+1)$-expression of the same graph. To each node $i$ of the tree decompositions, we associate three $(w+1)$-terminal graphs:
\begin{itemize}
\item $G^b_i = (B_i, B_i, E(G[B_i]))$, the $(w+1)$-terminal base graph corresponding to graph $G[B_i]$ induced by  bag $B_i$;
\item $G_i = (V_i, B_i, E(G[V_i]))$, corresponding to $G[V_i]$, with bag $B_i$ as set of terminals;
\item  If $i$ differs from the root, $G^+_i = (V_i \cup B_{p(i)}, B_{p(i)}, E(G[V_i \cup B_{p(i)} ]))$ corresponding to the graph induced by $V_i \cup B_{p(i)}$, with $B_{p(i)}$ as set of terminals.
\end{itemize}

Let us describe how to compute the $(w+1)$-expression of these graphs, by parsing bottom-up the tree decomposition (see also Figure~\ref{fig:twgram} applied to the tree decomposition of Figure~\ref{fig:ExTrDec}).

\begin{figure}[h!]
\centering
\scalebox{0.85}{
\input{Figures/ExTrGram2.tex}
}
\caption{From a tree decomposition of width $w$ to a $(w+1)$-expression.}
\label{fig:twgram}
\end{figure}
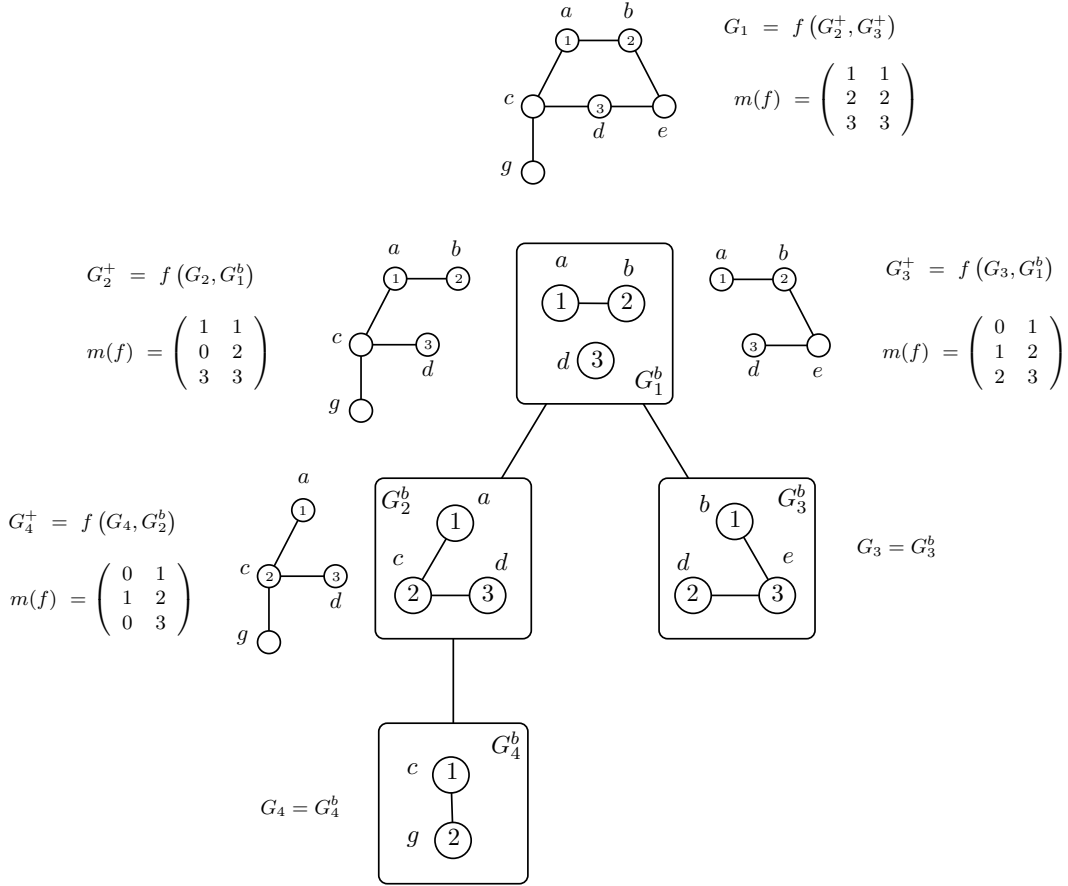

When $i$ is a leaf, $G_i = G^b_i$ is a $(w+1)$-terminal base graph.  Assume now that $i$ is not a leaf and let $Children(i)$ be the children of node $i$ in the decomposition tree. For each $j \in Children(i)$, we already possess an expression of the $(w+1)$-terminal graph $G_j = (V_j, B_j, E(G[V_j]))$. Observe that $G^+_j$ is obtained from a glueing of $G_j$ and the base graph $G^b_i$, where the terminals of $G_j$ contained in $B_j \cap B_i$ are glued on the corresponding terminals of $G^+_j$, and the others become non-terminals. Eventually, if $i$ has more than one child, then $G_i$ is obtained by the consecutive glueing of all $G^+_j$, $j \in Children(i)$, where the glueing is performed on $B_i$ by the same matrix $m(f)$ having $m_{r1}(f) = m_{r2}(f)= r$, for $1 \leq r \leq |B_i|$.

\subsection{Regular Properties and MSO}

We consider graph properties $\cP(G)$ assigning to each graph $G$ a boolean value. We have in mind properties expressible in Monadic Second Order Logic, like ``$G$ is not 3-colourable'', ``$G$ does not contain a given minor'', etc. Nevertheless, technically, we do not need the definition of MSO formulae, and the interested reader may refer to~\cite{CoEn12}; we only need the fact that MSO properties are \emph{regular}, in the sense defined below. By Courcelle's theorem, such properties can be decided in linear (sequential) time on graphs of bounded treewidth, if the tree decomposition (or the corresponding expression as a terminal recursive graph) is part of the input.

\begin{definition}[regular property]\label{de:reg}
A graph property $\cP$ is called \emph{regular} if, for any value $w$, we can associate a finite set $\cC$ of \emph{homomorphism classes} and a \emph{homomorphism function} $h$, assigning to each $w$-terminal recursive graph $G$ a class $h(G) \in \cC$ such that:
\begin{enumerate}
\item If $h(G_1) = h(G_2)$ then $\cP(G_1) = \cP(G_2)$.
\item For each composition operation $f$ of arity 2 there exists a function $\odot_f: \cC \times \cC \rightarrow \cC$ such that, for any two $w$-terminal recursive graphs $G_1$ and $G_2$,
 $$h(f(G_1,G_2)) = \odot_{f}(h(G_1),h(G_2))$$
and for each composition operation $f$ of arity 1 there is a function $\odot_{f}: \cC \rightarrow \cC$ such that, for any $w$-terminal recursive graph $G$,
$$h(f(G)) = \odot_{f}(h(G)).$$
\end{enumerate}
\end{definition}

We illustrate this definition on the property ``$G$ is not 3-colourable''. We can choose, as homomorphism  $h(G=(V,W,E))$, the set of all three-partitions $(W_1, W_2, W_3)$ of the set $W$ of terminals, such that graph $G$ has, as three colouring,  the one where each colour $i \in \{1,2,3\}$ intersects $W$ exactly in the set $W_i$. Observe that graph $G$ satisfies the 
property of not being 3-colourable if and only if its homomorphism class is the empty set. It is a matter of exercise to figure out how to compute the homomorphism 
class of a base $w$-terminal graph (by enumerating all its three-partitions into independent sets), and how to compute functions $\odot_f$ updating the class of the graph after a composition operation~$f$.

The first condition of Definition~\ref{de:reg} separates the classes into \emph{accepting} ones (i.e., classes $c \in \cC$ such that $h(G)=c$ implies that $\cP(G)$ is true) and \emph{rejecting} ones (i.e., classes $c \in \cC$ such that  $h(G)=c$ implies that $\cP(G)$ is false).
In full words, the second condition states that, if we perform a composition operation on two graphs (resp. one graph), the homomorphism class of the result can be obtained from the homomorphism classes of the graphs on which these operations are applied. Therefore, if a $w$-terminal recursive graph is given together with its expression in this grammar, and if moreover we know how to compute the homomorphism classes of the base graphs and the composition functions $\odot_f$ over all possible composition operations $f$, then the homomorphism class of the whole graph for a regular property $\cP$ can be obtained by dynamic programming. We simply need to parse the expression from bottom to top and, at each node, we compute the class of the corresponding sub-expression thanks to the second condition of regularity. At the root, the property is true if and only if we are in an accepting class. 

\begin{proposition}[\cite{BoPaTo92,Courcelle90}]\label{pr:reg}
Any property $\cP$ expressible by a $MSO$ formula is regular. Moreover, given the $MSO$ formula $\varphi$ and parameter $w$, one can explicitely compute the set of classes, the homomorphism function for all $w$-terminal base graphs as well as the composition functions $\odot_f$ of all possible composition operations $f$. 
\end{proposition}

Altogether, this provides an effective algorithm for checking property $\cP(G)$ in $O(n)$ time, by a sequential algorithm, given the $w$-expression (or, equivalently, the tree decomposition of width $w-1$) of the input graph, by computing bottom-up the homomorphism classes.

The notions of MSO and regular properties extend to properties on graphs and vertex subsets, i.e., we can consider properties $\cP(G,X)$ assigning to each graph $G$ and vertex subset $X$ of $G$ a boolean value. This allows to capture properties as ``$X$ is an independent set of $G$'', or  ``$X$ is an dominating set of $G$''. Moreover, the whole framework can capture the problem of computing a (or, in our case, certifying that) set $X$ is of maximum weight among those satisfying $\cP(G,X)$, for graphs with polynomial weights on their vertices. This issue is postponed to Appendix~\ref{app:propt}.

%% file: Figures/ExTrDec2.tex
\tikzset{every picture/.style={line width=0.75pt}} 

\begin{tikzpicture}[x=0.75pt,y=0.75pt,yscale=-1,xscale=1]

\draw    (337.5,147.14) -- (337.5,87.14) ;
\draw    (392.5,37.86) -- (337.5,87.14) ;
\draw    (392.5,37.86) -- (452.5,87.14) ;
\draw    (216.85,102.68) -- (170.68,102.68) ;
\draw    (192.41,55.15) -- (216.85,102.68) ;
\draw    (192.41,55.15) -- (147.59,55.15) ;
\draw    (170.68,102.68) -- (123.15,102.68) ;
\draw    (147.59,55.15) -- (123.15,102.68) ;
\draw    (123.15,102.68) -- (123.15,150.21) ;
\draw  [fill={rgb, 255:red, 255; green, 255; blue, 255 }  ,fill opacity=1 ] (162.53,102.68) .. controls (162.53,98.18) and (166.18,94.53) .. (170.68,94.53) .. controls (175.18,94.53) and (178.83,98.18) .. (178.83,102.68) .. controls (178.83,107.18) and (175.18,110.83) .. (170.68,110.83) .. controls (166.18,110.83) and (162.53,107.18) .. (162.53,102.68) -- cycle ;
\draw  [fill={rgb, 255:red, 255; green, 255; blue, 255 }  ,fill opacity=1 ] (115,150.21) .. controls (115,145.71) and (118.65,142.06) .. (123.15,142.06) .. controls (127.65,142.06) and (131.3,145.71) .. (131.3,150.21) .. controls (131.3,154.71) and (127.65,158.36) .. (123.15,158.36) .. controls (118.65,158.36) and (115,154.71) .. (115,150.21) -- cycle ;
\draw  [fill={rgb, 255:red, 255; green, 255; blue, 255 }  ,fill opacity=1 ] (115,102.68) .. controls (115,98.18) and (118.65,94.53) .. (123.15,94.53) .. controls (127.65,94.53) and (131.3,98.18) .. (131.3,102.68) .. controls (131.3,107.18) and (127.65,110.83) .. (123.15,110.83) .. controls (118.65,110.83) and (115,107.18) .. (115,102.68) -- cycle ;
\draw  [fill={rgb, 255:red, 255; green, 255; blue, 255 }  ,fill opacity=1 ] (139.44,55.15) .. controls (139.44,50.65) and (143.09,47) .. (147.59,47) .. controls (152.09,47) and (155.74,50.65) .. (155.74,55.15) .. controls (155.74,59.65) and (152.09,63.3) .. (147.59,63.3) .. controls (143.09,63.3) and (139.44,59.65) .. (139.44,55.15) -- cycle ;
\draw  [fill={rgb, 255:red, 255; green, 255; blue, 255 }  ,fill opacity=1 ] (208.7,102.68) .. controls (208.7,98.18) and (212.35,94.53) .. (216.85,94.53) .. controls (221.35,94.53) and (225,98.18) .. (225,102.68) .. controls (225,107.18) and (221.35,110.83) .. (216.85,110.83) .. controls (212.35,110.83) and (208.7,107.18) .. (208.7,102.68) -- cycle ;
\draw  [fill={rgb, 255:red, 255; green, 255; blue, 255 }  ,fill opacity=1 ] (184.26,55.15) .. controls (184.26,50.65) and (187.91,47) .. (192.41,47) .. controls (196.91,47) and (200.56,50.65) .. (200.56,55.15) .. controls (200.56,59.65) and (196.91,63.3) .. (192.41,63.3) .. controls (187.91,63.3) and (184.26,59.65) .. (184.26,55.15) -- cycle ;
\draw  [fill={rgb, 255:red, 255; green, 255; blue, 255 }  ,fill opacity=1 ] (370,30) .. controls (370,27.24) and (372.24,25) .. (375,25) -- (410,25) .. controls (412.76,25) and (415,27.24) .. (415,30) -- (415,50) .. controls (415,52.76) and (412.76,55) .. (410,55) -- (375,55) .. controls (372.24,55) and (370,52.76) .. (370,50) -- cycle ;
\draw  [fill={rgb, 255:red, 255; green, 255; blue, 255 }  ,fill opacity=1 ] (315,78.29) .. controls (315,75.52) and (317.24,73.29) .. (320,73.29) -- (355,73.29) .. controls (357.76,73.29) and (360,75.52) .. (360,78.29) -- (360,99) .. controls (360,101.76) and (357.76,104) .. (355,104) -- (320,104) .. controls (317.24,104) and (315,101.76) .. (315,99) -- cycle ;
\draw  [fill={rgb, 255:red, 255; green, 255; blue, 255 }  ,fill opacity=1 ] (430,79.29) .. controls (430,76.52) and (432.24,74.29) .. (435,74.29) -- (470,74.29) .. controls (472.76,74.29) and (475,76.52) .. (475,79.29) -- (475,100) .. controls (475,102.76) and (472.76,105) .. (470,105) -- (435,105) .. controls (432.24,105) and (430,102.76) .. (430,100) -- cycle ;
\draw  [fill={rgb, 255:red, 255; green, 255; blue, 255 }  ,fill opacity=1 ] (315,139.29) .. controls (315,136.52) and (317.24,134.29) .. (320,134.29) -- (355,134.29) .. controls (357.76,134.29) and (360,136.52) .. (360,139.29) -- (360,160) .. controls (360,162.76) and (357.76,165) .. (355,165) -- (320,165) .. controls (317.24,165) and (315,162.76) .. (315,160) -- cycle ;

\draw (142,30) node [anchor=north west][inner sep=0.75pt]    {$a$};
\draw (188,27) node [anchor=north west][inner sep=0.75pt]    {$b$};
\draw (103,90.4) node [anchor=north west][inner sep=0.75pt]    {$c$};
\draw (166,112) node [anchor=north west][inner sep=0.75pt]    {$d$};
\draw (212,116) node [anchor=north west][inner sep=0.75pt]    {$e$};
\draw (102,140.4) node [anchor=north west][inner sep=0.75pt]    {$g$};
\draw (373,32) node [anchor=north west][inner sep=0.75pt]    {$a,b,d$};
\draw (318,81) node [anchor=north west][inner sep=0.75pt]    {$a,c,d$};
\draw (325,145) node [anchor=north west][inner sep=0.75pt]    {$c,g$};
\draw (435,82) node [anchor=north west][inner sep=0.75pt]    {$b,d,e$};

\end{tikzpicture}

%% file: Figures/ExTrGram2.tex

\tikzset{every picture/.style={line width=0.75pt}} 

\begin{tikzpicture}[x=0.75pt,y=0.75pt,yscale=-1,xscale=1]

\draw    (378.67,223.5) -- (462.21,362.5) ;
\draw    (295.67,362.5) -- (295.67,507.5) ;
\draw    (378.67,223.5) -- (295.67,362.5) ;
\draw  [fill={rgb, 255:red, 255; green, 255; blue, 255 }  ,fill opacity=1 ] (249.9,320) .. controls (249.9,317.24) and (252.14,315) .. (254.9,315) -- (336.44,315) .. controls (339.2,315) and (341.44,317.24) .. (341.44,320) -- (341.44,405) .. controls (341.44,407.76) and (339.2,410) .. (336.44,410) -- (254.9,410) .. controls (252.14,410) and (249.9,407.76) .. (249.9,405) -- cycle ;
\draw    (296.81,341.63) -- (272.07,384.37) ;
\draw  [fill={rgb, 255:red, 255; green, 255; blue, 255 }  ,fill opacity=1 ] (416.44,320) .. controls (416.44,317.24) and (418.68,315) .. (421.44,315) -- (502.98,315) .. controls (505.74,315) and (507.98,317.24) .. (507.98,320) -- (507.98,405) .. controls (507.98,407.76) and (505.74,410) .. (502.98,410) -- (421.44,410) .. controls (418.68,410) and (416.44,407.76) .. (416.44,405) -- cycle ;
\draw    (485.81,384.37) -- (436.44,384.37) ;
\draw  [fill={rgb, 255:red, 255; green, 255; blue, 255 }  ,fill opacity=1 ] (332.9,181) .. controls (332.9,178.24) and (335.14,176) .. (337.9,176) -- (419.44,176) .. controls (422.2,176) and (424.44,178.24) .. (424.44,181) -- (424.44,266) .. controls (424.44,268.76) and (422.2,271) .. (419.44,271) -- (337.9,271) .. controls (335.14,271) and (332.9,268.76) .. (332.9,266) -- cycle ;
\draw    (397.27,211.63) -- (358.53,211.37) ;
\draw    (298.44,196.92) -- (261.59,196.92) ;
\draw    (280.58,236) -- (241.5,236) ;
\draw    (261.59,196.92) -- (241.5,236) ;
\draw    (241.5,236) -- (241.5,275.08) ;
\draw  [fill={rgb, 255:red, 255; green, 255; blue, 255 }  ,fill opacity=1 ] (273.88,236) .. controls (273.88,232.3) and (276.88,229.3) .. (280.58,229.3) .. controls (284.28,229.3) and (287.27,232.3) .. (287.27,236) .. controls (287.27,239.7) and (284.28,242.7) .. (280.58,242.7) .. controls (276.88,242.7) and (273.88,239.7) .. (273.88,236) -- cycle ;
\draw  [fill={rgb, 255:red, 255; green, 255; blue, 255 }  ,fill opacity=1 ] (234.8,275.08) .. controls (234.8,271.38) and (237.8,268.38) .. (241.5,268.38) .. controls (245.2,268.38) and (248.2,271.38) .. (248.2,275.08) .. controls (248.2,278.78) and (245.2,281.78) .. (241.5,281.78) .. controls (237.8,281.78) and (234.8,278.78) .. (234.8,275.08) -- cycle ;
\draw  [fill={rgb, 255:red, 255; green, 255; blue, 255 }  ,fill opacity=1 ] (234.8,236) .. controls (234.8,232.3) and (237.8,229.3) .. (241.5,229.3) .. controls (245.2,229.3) and (248.2,232.3) .. (248.2,236) .. controls (248.2,239.7) and (245.2,242.7) .. (241.5,242.7) .. controls (237.8,242.7) and (234.8,239.7) .. (234.8,236) -- cycle ;
\draw  [fill={rgb, 255:red, 255; green, 255; blue, 255 }  ,fill opacity=1 ] (254.89,196.92) .. controls (254.89,193.22) and (257.89,190.22) .. (261.59,190.22) .. controls (265.29,190.22) and (268.29,193.22) .. (268.29,196.92) .. controls (268.29,200.62) and (265.29,203.62) .. (261.59,203.62) .. controls (257.89,203.62) and (254.89,200.62) .. (254.89,196.92) -- cycle ;
\draw  [fill={rgb, 255:red, 255; green, 255; blue, 255 }  ,fill opacity=1 ] (291.74,196.92) .. controls (291.74,193.22) and (294.74,190.22) .. (298.44,190.22) .. controls (302.14,190.22) and (305.14,193.22) .. (305.14,196.92) .. controls (305.14,200.62) and (302.14,203.62) .. (298.44,203.62) .. controls (294.74,203.62) and (291.74,200.62) .. (291.74,196.92) -- cycle ;
\draw  [fill={rgb, 255:red, 255; green, 255; blue, 255 }  ,fill opacity=1 ] (368.64,245.37) .. controls (368.64,239.5) and (373.4,234.74) .. (379.27,234.74) .. controls (385.14,234.74) and (389.9,239.5) .. (389.9,245.37) .. controls (389.9,251.24) and (385.14,256) .. (379.27,256) .. controls (373.4,256) and (368.64,251.24) .. (368.64,245.37) -- cycle ;
\draw  [fill={rgb, 255:red, 255; green, 255; blue, 255 }  ,fill opacity=1 ] (347.9,211.37) .. controls (347.9,205.5) and (352.66,200.74) .. (358.53,200.74) .. controls (364.4,200.74) and (369.16,205.5) .. (369.16,211.37) .. controls (369.16,217.24) and (364.4,222) .. (358.53,222) .. controls (352.66,222) and (347.9,217.24) .. (347.9,211.37) -- cycle ;
\draw  [fill={rgb, 255:red, 255; green, 255; blue, 255 }  ,fill opacity=1 ] (386.64,211.63) .. controls (386.64,205.76) and (391.4,201) .. (397.27,201) .. controls (403.14,201) and (407.9,205.76) .. (407.9,211.63) .. controls (407.9,217.5) and (403.14,222.26) .. (397.27,222.26) .. controls (391.4,222.26) and (386.64,217.5) .. (386.64,211.63) -- cycle ;
\draw    (419.5,95) -- (381.53,95) ;
\draw    (399.4,55.92) -- (419.5,95) ;
\draw    (399.4,55.92) -- (362.55,55.92) ;
\draw    (381.53,95) -- (342.45,95) ;
\draw    (362.55,55.92) -- (342.45,95) ;
\draw    (342.45,95) -- (342.45,134.08) ;
\draw  [fill={rgb, 255:red, 255; green, 255; blue, 255 }  ,fill opacity=1 ] (374.83,95) .. controls (374.83,91.3) and (377.83,88.3) .. (381.53,88.3) .. controls (385.23,88.3) and (388.23,91.3) .. (388.23,95) .. controls (388.23,98.7) and (385.23,101.7) .. (381.53,101.7) .. controls (377.83,101.7) and (374.83,98.7) .. (374.83,95) -- cycle ;
\draw  [fill={rgb, 255:red, 255; green, 255; blue, 255 }  ,fill opacity=1 ] (335.76,134.08) .. controls (335.76,130.38) and (338.75,127.38) .. (342.45,127.38) .. controls (346.15,127.38) and (349.15,130.38) .. (349.15,134.08) .. controls (349.15,137.78) and (346.15,140.78) .. (342.45,140.78) .. controls (338.75,140.78) and (335.76,137.78) .. (335.76,134.08) -- cycle ;
\draw  [fill={rgb, 255:red, 255; green, 255; blue, 255 }  ,fill opacity=1 ] (335.76,95) .. controls (335.76,91.3) and (338.75,88.3) .. (342.45,88.3) .. controls (346.15,88.3) and (349.15,91.3) .. (349.15,95) .. controls (349.15,98.7) and (346.15,101.7) .. (342.45,101.7) .. controls (338.75,101.7) and (335.76,98.7) .. (335.76,95) -- cycle ;
\draw  [fill={rgb, 255:red, 255; green, 255; blue, 255 }  ,fill opacity=1 ] (355.85,55.92) .. controls (355.85,52.22) and (358.85,49.22) .. (362.55,49.22) .. controls (366.25,49.22) and (369.25,52.22) .. (369.25,55.92) .. controls (369.25,59.62) and (366.25,62.62) .. (362.55,62.62) .. controls (358.85,62.62) and (355.85,59.62) .. (355.85,55.92) -- cycle ;
\draw  [fill={rgb, 255:red, 255; green, 255; blue, 255 }  ,fill opacity=1 ] (412.8,95) .. controls (412.8,91.3) and (415.8,88.3) .. (419.5,88.3) .. controls (423.2,88.3) and (426.2,91.3) .. (426.2,95) .. controls (426.2,98.7) and (423.2,101.7) .. (419.5,101.7) .. controls (415.8,101.7) and (412.8,98.7) .. (412.8,95) -- cycle ;
\draw  [fill={rgb, 255:red, 255; green, 255; blue, 255 }  ,fill opacity=1 ] (392.7,55.92) .. controls (392.7,52.22) and (395.7,49.22) .. (399.4,49.22) .. controls (403.1,49.22) and (406.1,52.22) .. (406.1,55.92) .. controls (406.1,59.62) and (403.1,62.62) .. (399.4,62.62) .. controls (395.7,62.62) and (392.7,59.62) .. (392.7,55.92) -- cycle ;
\draw    (510.08,236.4) -- (472.12,236.4) ;
\draw    (489.98,197.32) -- (510.08,236.4) ;
\draw    (489.98,197.32) -- (453.14,197.32) ;
\draw  [fill={rgb, 255:red, 255; green, 255; blue, 255 }  ,fill opacity=1 ] (465.42,236.4) .. controls (465.42,232.7) and (468.42,229.7) .. (472.12,229.7) .. controls (475.82,229.7) and (478.82,232.7) .. (478.82,236.4) .. controls (478.82,240.1) and (475.82,243.1) .. (472.12,243.1) .. controls (468.42,243.1) and (465.42,240.1) .. (465.42,236.4) -- cycle ;
\draw  [fill={rgb, 255:red, 255; green, 255; blue, 255 }  ,fill opacity=1 ] (446.44,197.32) .. controls (446.44,193.62) and (449.44,190.63) .. (453.14,190.63) .. controls (456.84,190.63) and (459.84,193.62) .. (459.84,197.32) .. controls (459.84,201.02) and (456.84,204.02) .. (453.14,204.02) .. controls (449.44,204.02) and (446.44,201.02) .. (446.44,197.32) -- cycle ;
\draw  [fill={rgb, 255:red, 255; green, 255; blue, 255 }  ,fill opacity=1 ] (503.38,236.4) .. controls (503.38,232.7) and (506.38,229.7) .. (510.08,229.7) .. controls (513.78,229.7) and (516.78,232.7) .. (516.78,236.4) .. controls (516.78,240.1) and (513.78,243.1) .. (510.08,243.1) .. controls (506.38,243.1) and (503.38,240.1) .. (503.38,236.4) -- cycle ;
\draw  [fill={rgb, 255:red, 255; green, 255; blue, 255 }  ,fill opacity=1 ] (483.29,197.32) .. controls (483.29,193.62) and (486.29,190.63) .. (489.98,190.63) .. controls (493.68,190.63) and (496.68,193.62) .. (496.68,197.32) .. controls (496.68,201.02) and (493.68,204.02) .. (489.98,204.02) .. controls (486.29,204.02) and (483.29,201.02) .. (483.29,197.32) -- cycle ;
\draw    (315.81,384.37) -- (272.07,384.37) ;
\draw  [fill={rgb, 255:red, 255; green, 255; blue, 255 }  ,fill opacity=1 ] (305.18,384.37) .. controls (305.18,378.5) and (309.94,373.74) .. (315.81,373.74) .. controls (321.68,373.74) and (326.44,378.5) .. (326.44,384.37) .. controls (326.44,390.24) and (321.68,395) .. (315.81,395) .. controls (309.94,395) and (305.18,390.24) .. (305.18,384.37) -- cycle ;
\draw  [fill={rgb, 255:red, 255; green, 255; blue, 255 }  ,fill opacity=1 ] (261.44,384.37) .. controls (261.44,378.5) and (266.2,373.74) .. (272.07,373.74) .. controls (277.94,373.74) and (282.7,378.5) .. (282.7,384.37) .. controls (282.7,390.24) and (277.94,395) .. (272.07,395) .. controls (266.2,395) and (261.44,390.24) .. (261.44,384.37) -- cycle ;
\draw  [fill={rgb, 255:red, 255; green, 255; blue, 255 }  ,fill opacity=1 ] (286.18,341.63) .. controls (286.18,335.76) and (290.94,331) .. (296.81,331) .. controls (302.68,331) and (307.44,335.76) .. (307.44,341.63) .. controls (307.44,347.5) and (302.68,352.26) .. (296.81,352.26) .. controls (290.94,352.26) and (286.18,347.5) .. (286.18,341.63) -- cycle ;
\draw    (485.81,384.37) -- (460.81,340.63) ;
\draw  [fill={rgb, 255:red, 255; green, 255; blue, 255 }  ,fill opacity=1 ] (475.18,384.37) .. controls (475.18,378.5) and (479.94,373.74) .. (485.81,373.74) .. controls (491.68,373.74) and (496.44,378.5) .. (496.44,384.37) .. controls (496.44,390.24) and (491.68,395) .. (485.81,395) .. controls (479.94,395) and (475.18,390.24) .. (475.18,384.37) -- cycle ;
\draw  [fill={rgb, 255:red, 255; green, 255; blue, 255 }  ,fill opacity=1 ] (450.18,340.63) .. controls (450.18,334.76) and (454.94,330) .. (460.81,330) .. controls (466.68,330) and (471.44,334.76) .. (471.44,340.63) .. controls (471.44,346.5) and (466.68,351.26) .. (460.81,351.26) .. controls (454.94,351.26) and (450.18,346.5) .. (450.18,340.63) -- cycle ;
\draw  [fill={rgb, 255:red, 255; green, 255; blue, 255 }  ,fill opacity=1 ] (425.81,384.37) .. controls (425.81,378.5) and (430.57,373.74) .. (436.44,373.74) .. controls (442.31,373.74) and (447.07,378.5) .. (447.07,384.37) .. controls (447.07,390.24) and (442.31,395) .. (436.44,395) .. controls (430.57,395) and (425.81,390.24) .. (425.81,384.37) -- cycle ;
\draw  [fill={rgb, 255:red, 255; green, 255; blue, 255 }  ,fill opacity=1 ] (251.9,465) .. controls (251.9,462.24) and (254.14,460) .. (256.9,460) -- (334.44,460) .. controls (337.2,460) and (339.44,462.24) .. (339.44,465) -- (339.44,550) .. controls (339.44,552.76) and (337.2,555) .. (334.44,555) -- (256.9,555) .. controls (254.14,555) and (251.9,552.76) .. (251.9,550) -- cycle ;
\draw    (295.53,529.37) -- (294.27,490.63) ;
\draw  [fill={rgb, 255:red, 255; green, 255; blue, 255 }  ,fill opacity=1 ] (283.64,490.63) .. controls (283.64,484.76) and (288.4,480) .. (294.27,480) .. controls (300.14,480) and (304.9,484.76) .. (304.9,490.63) .. controls (304.9,496.5) and (300.14,501.26) .. (294.27,501.26) .. controls (288.4,501.26) and (283.64,496.5) .. (283.64,490.63) -- cycle ;
\draw  [fill={rgb, 255:red, 255; green, 255; blue, 255 }  ,fill opacity=1 ] (284.9,529.37) .. controls (284.9,523.5) and (289.66,518.74) .. (295.53,518.74) .. controls (301.4,518.74) and (306.16,523.5) .. (306.16,529.37) .. controls (306.16,535.24) and (301.4,540) .. (295.53,540) .. controls (289.66,540) and (284.9,535.24) .. (284.9,529.37) -- cycle ;
\draw    (226.53,373) -- (187.45,373) ;
\draw    (207.55,333.92) -- (187.45,373) ;
\draw    (187.45,373) -- (187.45,412.08) ;
\draw  [fill={rgb, 255:red, 255; green, 255; blue, 255 }  ,fill opacity=1 ] (219.83,373) .. controls (219.83,369.3) and (222.83,366.3) .. (226.53,366.3) .. controls (230.23,366.3) and (233.23,369.3) .. (233.23,373) .. controls (233.23,376.7) and (230.23,379.7) .. (226.53,379.7) .. controls (222.83,379.7) and (219.83,376.7) .. (219.83,373) -- cycle ;
\draw  [fill={rgb, 255:red, 255; green, 255; blue, 255 }  ,fill opacity=1 ] (180.76,412.08) .. controls (180.76,408.38) and (183.75,405.38) .. (187.45,405.38) .. controls (191.15,405.38) and (194.15,408.38) .. (194.15,412.08) .. controls (194.15,415.78) and (191.15,418.78) .. (187.45,418.78) .. controls (183.75,418.78) and (180.76,415.78) .. (180.76,412.08) -- cycle ;
\draw  [fill={rgb, 255:red, 255; green, 255; blue, 255 }  ,fill opacity=1 ] (180.76,373) .. controls (180.76,369.3) and (183.75,366.3) .. (187.45,366.3) .. controls (191.15,366.3) and (194.15,369.3) .. (194.15,373) .. controls (194.15,376.7) and (191.15,379.7) .. (187.45,379.7) .. controls (183.75,379.7) and (180.76,376.7) .. (180.76,373) -- cycle ;
\draw  [fill={rgb, 255:red, 255; green, 255; blue, 255 }  ,fill opacity=1 ] (200.85,333.92) .. controls (200.85,330.22) and (203.85,327.22) .. (207.55,327.22) .. controls (211.25,327.22) and (214.25,330.22) .. (214.25,333.92) .. controls (214.25,337.62) and (211.25,340.62) .. (207.55,340.62) .. controls (203.85,340.62) and (200.85,337.62) .. (200.85,333.92) -- cycle ;

\draw (256,175) node [anchor=north west][inner sep=0.75pt]  [font=\small]  {$a$};
\draw (293,173) node [anchor=north west][inner sep=0.75pt]  [font=\small]  {$b$};
\draw (222.86,227.55) node [anchor=north west][inner sep=0.75pt]  [font=\small]  {$c$};
\draw (275.66,244) node [anchor=north west][inner sep=0.75pt]  [font=\small]  {$d$};
\draw (221.04,266.66) node [anchor=north west][inner sep=0.75pt]  [font=\small]  {$g$};
\draw (353.9,182.5) node [anchor=north west][inner sep=0.75pt]    {$a$};
\draw (393.9,183.5) node [anchor=north west][inner sep=0.75pt]    {$b$};
\draw (354.9,238.5) node [anchor=north west][inner sep=0.75pt]    {$d$};
\draw (354,204) node [anchor=north west][inner sep=0.75pt]    {$1$};
\draw (392.5,204) node [anchor=north west][inner sep=0.75pt]    {$2$};
\draw (375,238) node [anchor=north west][inner sep=0.75pt]    {$3$};
\draw (258,193) node [anchor=north west][inner sep=0.75pt]  [font=\tiny]  {$1$};
\draw (295,193) node [anchor=north west][inner sep=0.75pt]  [font=\tiny]  {$2$};
\draw (277,232) node [anchor=north west][inner sep=0.75pt]  [font=\tiny]  {$3$};
\draw (357,35) node [anchor=north west][inner sep=0.75pt]  [font=\small]  {$a$};
\draw (394,32) node [anchor=north west][inner sep=0.75pt]  [font=\small]  {$b$};
\draw (323.82,86.55) node [anchor=north west][inner sep=0.75pt]  [font=\small]  {$c$};
\draw (376,103) node [anchor=north west][inner sep=0.75pt]  [font=\small]  {$d$};
\draw (414,106) node [anchor=north west][inner sep=0.75pt]  [font=\small]  {$e$};
\draw (322,125.66) node [anchor=north west][inner sep=0.75pt]  [font=\small]  {$g$};
\draw (359,52) node [anchor=north west][inner sep=0.75pt]  [font=\tiny]  {$1$};
\draw (396,52) node [anchor=north west][inner sep=0.75pt]  [font=\tiny]  {$2$};
\draw (378,92) node [anchor=north west][inner sep=0.75pt]  [font=\tiny]  {$3$};
\draw (448,178) node [anchor=north west][inner sep=0.75pt]  [font=\small]  {$a$};
\draw (485,175) node [anchor=north west][inner sep=0.75pt]  [font=\small]  {$b$};
\draw (467,244) node [anchor=north west][inner sep=0.75pt]  [font=\small]  {$d$};
\draw (505.03,248) node [anchor=north west][inner sep=0.75pt]  [font=\small]  {$e$};
\draw (450,193) node [anchor=north west][inner sep=0.75pt]  [font=\tiny]  {$1$};
\draw (486,193) node [anchor=north west][inner sep=0.75pt]  [font=\tiny]  {$2$};
\draw (468,233) node [anchor=north west][inner sep=0.75pt]  [font=\tiny]  {$3$};
\draw (308.44,322.4) node [anchor=north west][inner sep=0.75pt]    {$a$};
\draw (258.44,358.4) node [anchor=north west][inner sep=0.75pt]    {$c$};
\draw (318.44,357.4) node [anchor=north west][inner sep=0.75pt]    {$d$};
\draw (292,334) node [anchor=north west][inner sep=0.75pt]    {$1$};
\draw (267,377) node [anchor=north west][inner sep=0.75pt]    {$2$};
\draw (311,377) node [anchor=north west][inner sep=0.75pt]    {$3$};
\draw (487.44,357.4) node [anchor=north west][inner sep=0.75pt]    {$e$};
\draw (425.44,357.4) node [anchor=north west][inner sep=0.75pt]    {$d$};
\draw (438.44,323.4) node [anchor=north west][inner sep=0.75pt]    {$b$};
\draw (456,333) node [anchor=north west][inner sep=0.75pt]    {$1$};
\draw (431,377) node [anchor=north west][inner sep=0.75pt]    {$2$};
\draw (481,377) node [anchor=north west][inner sep=0.75pt]    {$3$};
\draw (266.9,482.4) node [anchor=north west][inner sep=0.75pt]    {$c$};
\draw (266.9,522.4) node [anchor=north west][inner sep=0.75pt]    {$g$};
\draw (289.9,482.4) node [anchor=north west][inner sep=0.75pt]    {$1$};
\draw (290.9,521.4) node [anchor=north west][inner sep=0.75pt]    {$2$};
\draw (202.89,310.4) node [anchor=north west][inner sep=0.75pt]  [font=\small]  {$a$};
\draw (168.82,364.55) node [anchor=north west][inner sep=0.75pt]  [font=\small]  {$c$};
\draw (221.62,381) node [anchor=north west][inner sep=0.75pt]  [font=\small]  {$d$};
\draw (167,403.66) node [anchor=north west][inner sep=0.75pt]  [font=\small]  {$g$};
\draw (203.5,330) node [anchor=north west][inner sep=0.75pt]  [font=\tiny]  {$1$};
\draw (223,369) node [anchor=north west][inner sep=0.75pt]  [font=\tiny]  {$3$};
\draw (184,369) node [anchor=north west][inner sep=0.75pt]  [font=\tiny]  {$2$};
\draw (181,501.4) node [anchor=north west][inner sep=0.75pt]  [font=\footnotesize]  {$G_{4} =G_{4}^{b}$};
\draw (316.44,463.4) node [anchor=north west][inner sep=0.75pt]    {$G_{4}^{b}$};
\draw (251.9,318.4) node [anchor=north west][inner sep=0.75pt]    {$G_{2}^{b}$};
\draw (484.44,317.4) node [anchor=north west][inner sep=0.75pt]    {$G_{3}^{b}$};
\draw (400.44,247.4) node [anchor=north west][inner sep=0.75pt]    {$G_{1}^{b}$};
\draw (531,347.4) node [anchor=north west][inner sep=0.75pt]  [font=\footnotesize]  {$G_{3} =G_{3}^{b}$};
\draw (33,332.4) node [anchor=north west][inner sep=0.75pt]  [font=\footnotesize]  {$G_{4}^{+} \ =\ f\left( G_{4} ,G_{2}^{b}\right)$};
\draw (34,363.4) node [anchor=north west][inner sep=0.75pt]  [font=\footnotesize]  {$m( f) \ =\left(\begin{array}{ c c }
0 & 1\\
1 & 2\\
0 & 3
\end{array}\right)$};
\draw (79,186.4) node [anchor=north west][inner sep=0.75pt]  [font=\footnotesize]  {$G_{2}^{+} \ =\ f\left( G_{2} ,G_{1}^{b}\right)$};
\draw (79,217.4) node [anchor=north west][inner sep=0.75pt]  [font=\footnotesize]  {$m( f) \ =\left(\begin{array}{ c c }
1 & 1\\
0 & 2\\
3 & 3
\end{array}\right)$};
\draw (548,182.4) node [anchor=north west][inner sep=0.75pt]  [font=\footnotesize]  {$G_{3}^{+} \ =\ f\left( G_{3} ,G_{1}^{b}\right)$};
\draw (546,217.4) node [anchor=north west][inner sep=0.75pt]  [font=\footnotesize]  {$m( f) \ =\left(\begin{array}{ c c }
0 & 1\\
1 & 2\\
2 & 3
\end{array}\right)$};
\draw (453,39.4) node [anchor=north west][inner sep=0.75pt]  [font=\footnotesize]  {$G_{1} \ =\ f\left( G_{2}^{+} ,G_{3}^{+}\right)$};
\draw (459,67.4) node [anchor=north west][inner sep=0.75pt]  [font=\footnotesize]  {$m( f) \ =\left(\begin{array}{ c c }
1 & 1\\
2 & 2\\
3 & 3
\end{array}\right)$};

\end{tikzpicture}

%% file: prep.tex

\section{Preparing the Tree Decompositions}

By a classic result of Bodlaender~\cite{Bodlaender88}, an optimal tree decomposition of graph $G$ can be transformed into a decomposition whose tree is of logarithmic depth, while the size of the bags is at most multiplied by $3$. We strongly rely on such decomposition, plus a connectivity property that we call \emph{coherence}. We say that a rooted tree decomposition of 
a graph $G=(V,E)$ is \emph{coherent} if for every $i\in I$, the set $F_i$ is non empty and the graph $G[V_i \setminus B_{p(i)}]$ is connected.

\begin{lemma}\label{le:logco}
Let $k \geq 1$, and let $G$ be a connected $n$-vertex graph of treewidth at most $k$. Then, $G$ admits a coherent tree decomposition of width at most $3k+2$ and depth $\cO(\log n)$.
\end{lemma}
\begin{proof}

Firstly, choose a tree decomposition $(T=(I,F), \{B_i, i \in I\})$ of $G$ where $T$ is of logarithmic depth and each bag is of size at most $3k+3$. Such a decomposition exists by~\cite{Bodlaender88}. 

Let us show how to transform this tree decomposition into a coherent one. Firstly we focus on the connectivity condition. Assume the decomposition is not coherent, and let $i$ a node that violates the connectivity condition, closest to the root. Observe that $i$ is different from the root: we can assume w.l.o.g that all bags are non empty, in particular $F_r$ is non empty and $G[V_r] = G$ is connected. Then suppose that $G[V_i \setminus B_{p(i)}]$ is not connected and let $W^1, W^2,\dots, W^p$ be the vertex sets of the connected components of $G[V_i \setminus B_{p(i)}]$. Denote by $N^j$ the neighbourhood of $W^j \in G$, for $1 \leq j \leq p$, and observe that every $N^j$ is a subset of $B_i \cap B_{p(i)}$. For each $j$, $1 \leq j \leq p$, construct a tree decomposition $T^j_i$ of $G[W^j \cup N^j]$ by taking a copy of $T_i$ and the corresponding bags, then restricting the bags to their intersection with  $W^j \cup N^j$. Eventually replace, in the tree decomposition $T$ of $G$, the subtree $T_i$ by the $p$ copies $T^1_i, \dots, T^p_i$, by making their roots adjacent to $p(i)$. Observe that we obtain indeed a tree decomposition of $G$, and the connectivity condition on node $i$ has been mended, in the sense that the $p$ new nodes corresponding to copies of node $i$ satisfy it: for each copy $i^j$ of $i$, we have that $V_{i^j} \setminus B_{p(i^j)} = W^j$. 

Now that the connectivity condition is satisfied for every node, we fix the condition stating that sets $F_i$ must not be empty. If $F_i$ is empty for some node $i$, then $B_i$ is a subset of $B_{p(i)}$. Therefore we can remove node $i$ from the decomposition tree and attach its children directly to the parent $p(i)$ of the deleted node, obtaining a new tree decomposition, without increasing the depth. 
This process can be iterated as long as necessary, hence we my assume that for any node $i$, $F_i$ is non empty. Also observe that the removal of node $i$ does not modify the set $V_j \setminus B_{p(j)}$ for any child $j$ of $i$ in the initial tree $T$, therefore the connectivity condition is preserved for all nodes.
\end{proof}

In our protocol we must be able to communicate, for any node $i$ of the decomposition, some information about $V_i$ to a vertex in the bag corresponding to the parent node $p(i)$, more precisely, to some vertex of $F_{p(i)}$. The following lemma shows the existence a vertex $\ell_i \in V_i \setminus B_{p(i)}$ adjacent in $G$ to some vertex $w \in F_{p(i)}$. 

\begin{lemma}\label{le:leaders}
Let $T=(I,F)$ be a coherent tree decomposition of $G=(V,E)$. Then, for every $i \in I$ different from the root there exists a pair of vertices  $\ell_i \in V_i \setminus B_{p(i)}$ and $w \in F_{p(i)}$ such that $\{w,\ell_i\}\in E$. 

Vertex $\ell_i$ is called the \emph{exit vertex of $i$}, and  $w$ is called the \emph{vertex of $F_{p(i)}$ in charge of node $i$}.
\end{lemma}
\begin{proof}
Denote $W_i = V_i \setminus B_{p(i)}$. By definition of a tree decomposition the neighbourhood $N_G(W_i)$ of $W_i$ in graph $G$ is a subset of $B_{p(i)}$. We must show that $N_G(W_i)$ contains at least one vertex $w$ in $F_{p(i)}$, which allows to take a $\ell_i \in W_i$ adjacent to $w$ in $G$. Assume by contradiction that $N_G(W_i)$ does not intersect $F_{p(i)}$. In this case $p(i)$ is not the root vertex, and $N_G(W_i) \subseteq B_{p(i)} \setminus F_i$, which is also equal to $B_{p(i)} \cap B_{p(p(i))}$. Therefore $B_{p(i)} \cap B_{p(p(i))}$ separates $W_i$ from the $F_{p(i)}$ in graph $G$. This contradicts the coherence of the tree decomposition, more precisely the connectivity property at node $p(i)$, since $W_i$ and $F_{p(i)}$ are disconnected in $G[V_{p(i)} \setminus B_{p(p(i))}]$.
\end{proof}

In our certification protocols, for each node $i$ of the decomposition tree, the vertices of $F_i$ as well as the exit vertex $\ell_i$ will receive from the prover some information concerning  graph $G_i = G[V_i]$. We will need to ensure that $\ell_i$ and all vertices of $F_i$ received the same information. For this purpose we use trees contained in $G[V_i \setminus B_{p(i)}]$, spanning $\ell_i$ and $F_i$.

\begin{lemma}\label{le:SteinerCongestion}
Consider a coherent tree decomposition $T=(I,F)$ of graph $G=(V,E)$, of depth $O(\log n)$. For each node $i$ of the decomposition tree, there is a subtree $S(i)$ of $G[V_i \setminus B_{p(i)}]$ spanning $F_i$ and the exit vertex $\ell_i$.

Moreover each vertex of $G$ appears $O(\log n)$ times in the family of trees $\cT(G) = \{S(i) \mid i \in I\}$.
\end{lemma}
\begin{proof}
Since our tree decomposition is coherent, for each node $i$ graph $G[V_i \setminus B_{p(i)}]$ is connected so it contains the required subtree $S(i)$. 

Observe is that, if $i$ and $j$ are different nodes of the tree decomposition such that none is ancestor of the other, then sets $V_i \setminus B_{p(i)}$ and $V_j \setminus B_{p(j)}$ are disjoint, by definition of tree decompositions. Therefore, if we fix a vertex $v$ of $G$, the nodes $i$ such that $v$ appears in $S(i)$ are pairwise comparable w.r.t the ancestor relation in the decomposition tree. The decomposition tree is of depth $O(\log n)$ and the conclusion follows.
\end{proof}

%% file: twApprox.tex
\section{A Protocol Certifying a 3-Approximation of the Treewidth}\label{se:3approx}

In this section we describe a protocol certifying a 3-approximation of treewidth. More precisely, we prove the following Lemma.

\begin{lemma}\label{lem:3approx}  For each $k\geq 1$ there is a distributed certification protocol that uses messages of size $O(k^2 \log^2 n)$  and ensures, for any input graph $G$, that:
\[
\left\{\begin{array}{lcl}
\tw(G)\leq k & \Rightarrow & \mbox{there exists a certificate assignment s.t. all nodes accept;}\\ 
\tw(G) > 3k+2 & \Rightarrow & \mbox{for every certificate assignment, at least one node rejects.}\\ 
\end{array}\right.
\] 

\end{lemma}

Let us describe the messages that the prover sends to each vertex of $G$, if $\tw(G) \leq k$. These messages describe a coherent tree decomposition of width at mots $3k+2$ and of logarithmic depth, which exists by Lemma~\ref{le:logco}.

We identify node $i$ of the decomposition tree with the number  corresponding to a binary representation of the set of vertices $B_i$ contained in its bag, so $1\leq i \leq n^{\cO(k)}$. In full words, a node is simply identified by the content of its bag, which is possible since coherent tree decompositions have pairwise disjoint bags.

Our protocol distinguishes two types of certificates, namely \emph{main messages} and \emph{auxiliary messages}. Each vertex receives one main messages and $\cO(\log n)$ auxiliary messages. Let us describe each one of them.

\paragraph{Main messages.} These messages are used to encode a tree decomposition, following Definition~\ref{de:treedec}. Each vertex $v$ receives as a certificate the following messages, that we denote $m(v)$:
\begin{enumerate}
\item A number $d = d(v)$, representing the depth of the node $i$ such that $v \in F_i$
\item\label{it:bags} A list of sets $\cB(v) = B_{d}(v), B_{d-1}(v), \dots, B_1(v)$, representing the path of bags from node $i=B_d(v)$ to the root node. 
\item The list of sets $\cF(v) = F_{d}(v), F_{d-1}(v), \dots, F_1(v)$, representing the sets $F_j(v) = B_j(v) \setminus B_{j-1}(v)$, for each $j \in \{1, \dots, d\}$.
\item A list of sets $\cE(v) = E_d(v), \dots, E_1(v)$,  where,  for each $j\in \{1, \dots, d\}$, $E_j(v) \subseteq {B_j(v) \choose 2}$ represents the edge set of $G[B_j(v)]$.
\end{enumerate}

Observe that the size of a main message is $\cO(k^2 \log^2 n)$. 

 \paragraph{Auxiliary messages.} These messages allow to check the consistency of the main messages between vertices of a same set $F_i$, for each node $i$ of the decomposition.

From Lemma \ref{le:SteinerCongestion}, we have that for each node $i$ there is a subtree $S(i)$ connecting all pair of vertices of $F_i$ and the exit vertex $\ell_i$.  The vertices $w$ of $S(i)$ are called \emph{auxiliary vertices for $i$}.  For a vertex $w$, let  us call $Aux(w)$ the set of nodes $i$ such that $w$ is an auxiliary vertex for $i$.  From Lemma \ref{le:SteinerCongestion}, we know that for each $w\in V$, $|Aux(w)| = \cO(\log n)$.  

Each node $w$ receives the set $Aux(w)$ and for each $i \in Aux(w)$ the message $m_{aux}(w,i)$ containing the following information where 

\begin{itemize}
\item $d_{aux}(w,i)$ is the depth of node $i$.
\item $\ell_i(w)$ is a vertex identifier of the exit vertex of $F_i$ (cf. Lemma~\ref{le:leaders}).
\item $\alpha_i(w)$ is a vertex identifier of the vertex in $F_{p(i)}$ in charge of $B_i$  (cf. Lemma~\ref{le:leaders}).
\item $F_i(w)$ is a set of vertices, representing $F_i$. 
%
%
%
%
\item $TreeCert(w)$ is the certificate that receives $w$ in the protocol used to verify that  $S(i)$ is a tree rooted at $\ell_i$ and spanning $F_i(w)$. More precisely $cert(F_i,w) = (parent(w), dist(w), sub(w))$, where: 
\begin{itemize} 
\item $parent(w)$ represent the parent of $w$ in $S(i)$ ($parent(w) = \perp$ if $w = \ell_i(w)$), 
\item $dist(w)$ represents the distance from $w$ to $\ell_i$ in $S(i)$, and 
\item $sub(w)$ represents is the subset of $F_i(v)$ that are descendants of $w$ in $S(i)$. 
\end{itemize}

 \end{itemize}

Observe that for any given vertex $w$ and node $i$, the messages $m_{aux}(w,i)$ is of size $O(k \log n)$. Thanks to Lemma~\ref{le:SteinerCongestion}, a vertex $w$ appears $\cO(\log n)$ times as auxiliary vertex of some node $i$. Therefore, a vertex $w$ receives in total $O(k \log^2 n)$ bits for auxiliary messages. 

\medskip

\paragraph{Verification round.}

Given two vertices $u$ and $v$ such that $d(u)\leq d(v)$, we say that the main message of $u$ is a $d$-suffix of the main message of $v$ if $B_j(u) = B_j(v)$ and $E_j(u) = E_j(v)$ for each $j \in \{1, \dots, d\}$.

Let $d = d(v)$. In the verification round, vertex $v$ verifies the following conditions.
 
\paragraph{Consistency of the tree decomposition.} 
\begin{enumerate}
\item The size of each $B \in \cB(v)$ is at most $3k+3$.  \label{cond:1}
\item The set $F_d(v)$ contains $v$. \label{cond:1,2}
\item For each $j \in \{2, \dots, d\}$,  the set $F_{j}(v)$ equals  $B_{j}(v) \setminus B_{j-1}(v)$.  \label{cond:2}
\item For each $w \in V(G)$ and $j_1,j_2  \in \{1, \dots, d\}$ with $j_1 < j_2$, if $w \in B_{j_1} \cap B_{j_2}$, then $w \in B_{i_j}$  for every $j\in \{j_1+1, \dots, j_2-1\}$.  \label{cond:3}

\item For each $j_1, j_2 \in \{1, \dots, d\}$, each pair of vertices $u_1, u_2 \in B_{j_1}(v)\cap B_{j_2}(v)$ satisfies that $\{u_1,u_2\} \in E_{j_1}(v) \iff \{u_1, u_2\} \in E_{j_2}(v)$. \label{cond:4}
\item For each $u \in B_d(v)$, $v$ checks that $\{u,v\} \in E \iff \{u,v\}  \in E_d(v)$. \label{cond:5}
\item For each $u \in N(v)$ such that $d(u)\geq d(v)$, $v$ checks that $m(v)$ is a $d(v)$-suffix of $m(u)$.  \label{cond:6}
\item For each $u \in N(v)$ such that $d(u) \leq d(v)$, $v$ checks that  $u \in B_{d}(v)$.  \label{cond:7}
\item $v$ checks that it is an auxiliary vertex for $B_d(v)$ and that it has a neighbor that is also an auxiliary vertex for $ B_d(v)$.  \label{cond:8}
\item For each vertex $w\in N(v) \cup \{v\}$ such that $w$ is an auxiliary tree vertex for $B_d(v)$, $v$ checks that $d_{aux}(w, B_d(w)) = d$ and $F_i(w) = F_d(v)$. \label{cond:9}
\end{enumerate}

\paragraph{Consistency of the auxiliary trees and the exit vertex.}  The following conditions are used to verify that the nodes marked as auxiliary vertices for node $i$ form an auxiliary subtree $S(i)$ rooted at $\ell_i$ and spanning $F_i$. At the same time, we check that all de nodes in $S(i)$ have the same auxiliary information, corresponding to the depth $d_i$ of bag $i$,  the contents of $F_i$, the identity of exit vertex $\ell_i$, and the identity of the node of $F_{p(i)}$ responsable of $i$, and the same $d_i$-suffix of the main messages. 

For each $i \in Aux(v)$, vertex $v$ checks te following conditions
\begin{enumerate}
\setcounter{enumi}{10}
\item For each vertex $w\in N(v)$ such that $w$ is an auxiliary tree vertex for $i$, $v$ checks that $$(d_{aux}(w,i), \ell_i(w), \alpha_i(w), F_i(w)) =  (d_{aux}(v,i), \ell_i(v),  \label{cond:10}\alpha_i(v), F_i(v)) $$
\item $d_{aux}(v,i) \leq d(v)$.\label{cond:11}
\item  Uses $TreeCert(F_i(v), v)$ to verify that there is an auxiliary tree $S(i)$ rooted in $\ell_i(v)$ and spanning $F_i(v)$. More precisely, $v$ checks the following conditions:\label{cond:12}
\begin{enumerate}
\item If $v\neq \ell_i(v)$ then $v$ has a neighbor with the label $parent(w)$ which is also an auxiliary vertex for $i$;
\item If $v\neq \ell_i(v)$, then $dist(parent(v)) = dist(v) -1$;
\item If $v = \ell_i$ then $dist(v) = 0$, $sub(v) = F_i(v)$,  $v$ is adjacent to $\alpha_i(v)$ and $d(\alpha_i(v)) = d_{aux}(v,i)-1$.
\item Set $sub(v)$ is  the union of all sets $sub(w)$ over the children $w$ of $v$ in $S(i)$ (i.e., for all $w$ such that $parent(w) = v$), plus vertex $v$ itself if $v \in F_i$.
%
%
%
%
\end{enumerate}
\end{enumerate}

\paragraph{Soundness and completeness.} We now analyze the correctness of the protocol. The completeness follows directly by Lemmas \ref{le:logco}, \ref{le:leaders} and \ref{le:SteinerCongestion}. In the following, we prove the soundness. 

\paragraph{Soundness:} Let us assume that all vertices accept a given certificate in the verification round. We now show that necessarily $\tw(G)\leq 3k+2$. For each node $v\in V$, let us call $B(v)$ and $F(v)$ the set $F_{d(v)}(v)$ and $B_{d(v)}(v)$, respectively. We say that a vertex $v$ is in depth $d$ if $d(v) = d$. The proof of the soundness is a consequence of the following claims. 

\medskip

\noindent{\bf Claim 1:}  For each $i\in Aux(v)$,  there is a tree $S(i)$ rooted in $\ell_i(v)$ spanning~$F_i(v)$.  Moreover, all the vertices in $S(i)$ are in a depth greater or equal than $d_{aux}(v,i)$, and their main messages have the same $d_{aux}(v,i)$-suffix.

\begin{proof}[Proof of Claim 1] 

 First, observe that by the verification of condition {\bf \ref{cond:12} (a)-(c)}, we have that $S(i)$ is defined by the set of all auxiliary vertices for $i$ and the edges $\{w, parent(w)\}$. Since $S(i)$ is connected, by conditions {\bf \ref{cond:9}} and {\bf \ref{cond:10}}, all auxiliary vertices for node $i$ agree in the same $F_i = F_i(v)$ and in the depth of $i$ given by $d_{aux} =d_{aux}(v,i)$. By condition {\bf \ref{cond:12} (c)-(d)}, all vertices in $F_i$ exist and are auxiliary vertices for node $i$. Finally, by condition  {\bf \ref{cond:11}} all nodes are in a depth greater or equal than $d_{aux}$ and by condition {\bf \ref{cond:6}}, the main messages of all vertices in $S(i)$ have the same $d_{aux}$-suffix.
\end{proof}

\noindent{\bf Claim 2:} For every vertex $v$, all nodes in $F(v)$ receive the same main messages as $v$.

\begin{proof}[Proof of Claim 2] Let $u$ be a vertex in $F(v)$. If $u$ and $v$ are adjacent the claim is true by condition {\bf \ref{cond:6}}.  Suppose then that $u\notin N(v)$. Since $v$ verifies condition {\bf \ref{cond:8}}, there is a set of auxiliary vertices for node $i = B(v)$. By {\bf Claim 1}, $m(v)$ is a $d(v)$-suffix of $m(w)$, for every auxiliary vertex $w$ for node $i$. Since all vertices in $F(v)$ are auxiliary vertices for $i$, we deduce that $u$ has the same main messages than~$v$.  
\end{proof}

\noindent{\bf Claim 3:} For every pair of vertices $u, v\in V$ either $F(v) = F(u)$ or $F(v)\cap F(u) = \emptyset$. 

\begin{proof}[Proof of Claim 3]
This is a direct corollary of {\bf Claim 2}. Indeed, let us suppose that there exist a pair $u,v \in V$  such that $F(v) \neq F(u)$ but $F(v) \cap F(u) \neq \emptyset$. Then, without loss of generality, there is a node $w\in F(v)\cap F(u)$ such that $F(w) \neq F(v)$, which contradicts {\bf Claim~2}.
\end{proof}

\noindent{\bf Claim 4:} For every vertex $v$ such that $d(v)>1$, there exist a node $u$ such that $m(u)$ is a $(d(v)-1)$-suffix of $m(v)$.

\begin{proof}[Proof of Claim 4] 
Let $d=d(v)$. {\bf Claim 1} implies that the exit vertex $\ell_i$ for $i = B_d(v)$ exists and is the root of $S(i)$, which is in a depth greater or equal than $d_{aux} = d$. Condition {\bf \ref{cond:12} (c)} implies that $\ell_i$ is adjacent to a node $\alpha_i$ of depth $d - 1$. Then, by condition {\bf \ref{cond:6}}, $m(\alpha_i)$ is a $d-1$-suffix of $m(\ell_i)$. Since $m(v)$ is a $d$-suffix of $m(\ell_i)$, we deduce that $m(\alpha_i)$ is a $d-1$-suffix of $m(v)$.
\end{proof}

\medskip

\noindent{\bf Claim 5}: For every $u,v \in V$, the sets $F(u) \neq F(v)$ if and only if $B(u) \neq B(v)$. 
\begin{proof}[Proof of Claim 5]
First, observe that if $F(u) = F(v)$, then by condition {\bf \ref{cond:1,2}} and {\bf Claim 2}, $B(v) = B(u)$. For the reciprocal, let us suppose by contradiction that there exist $u,v\in V$ such that $F(u) \neq F(v)$ and $B(u) = B(v)$. Let us call $d_1 = d(u)$ and $d_2 = d(v)$. 
Since $F(u) \neq F(v)$, necessarily $B_{d_1-1}(u) \neq B_{d_2-1}(v)$. Let us assume, without loss of generality, that there exists a vertex $w \in F(v) \setminus F(u)$. 
Since $w$ belongs to $F(v)$, we have that $F(w) = F(v)$ by {\bf Claim 2}, and $w$ does not belong to $B_{d_1-1}(v)$. Since $w \notin F(u)$ we have that $w$ belongs to $B_{d_2-1}(u)$.
 Let us call $d_3$ the maximum in $\{1, \dots, d_1-1\}$ such that $B_{d_3}(u)$ belongs to $\cB(v)$. Observe that $d_3$ exists, because applying condition {\bf \ref{cond:6}} on all the vertices in $G$ we deduce that $B_1(u) = B_1(v)$. If $B_{d_3}(u)$ contains $w$, then $v$ fails to verify condition {\bf \ref{cond:3}}. If $B_{d_3}(u)$ does not contain vertex $w$, there exists a $d_4 \in \{d_1, \dots, d_3-1\}$ such that $w \in F_{d_4}(u) =  B_{d_4}(u) \setminus B_{d_4-1}(u)$. Then, {\bf Claim 4} applied to the vertices in the sequence $F_{d_1}(u), F_{d_1-1}(u), ..., F_{d_4}(u)$ implies that there is a node $w'$ such that $F(w') = F_{d_4}(u)$. Then, by {\bf Claim 2}, $F(w) = F_{d_4}(u)$. We deduce that $B(v) = B_{d_4}(u)$, which is a contradiction with the choice of $d_3$.
 \end{proof}

Let us define $I$ as the set of indexes  $i \in [n^{\cO(k)}]$ for which there is a $v \in V(G)$ such that $i$ is the binary representation of $B(v)$. By {\bf Claim 2}, {\bf 3} and {\bf 5}, we have a partition $\{F_i\}_{i \in I}$ of $V(G)$, such that, for each $i \in I$, all nodes in $F_i$ receive the same main messages. In particular, for every vertex $v$ in $F_i$, we have that $i$ is the binary representation of $B(v)$. For each $v\in F_i$, we define $p(i)$ as the binary representation of $B_{d(v)-1}(v)$ ($p(i)= \bot$ if $v \in B_1(v)$). From {\bf Claim~4} we know that the binary representation of $B_{d(v)-1}(v)$ is also in $I$. In other words, the nodes in $F_{d(v)-1}(v)$ have certificates that are consistent with the certificate of $v$. In particular, all vertices of $G$ agree on the contents of the root node, that we call $B_1$. We then define the pair $(T, \{B_i\}_{i \in I})$, where $T$ is defined by the tree with vertex set $I$ and edge set $\{i, p(i)\}$, for each $i\in I$ different than the root. \\

\medskip

\noindent{\bf Claim 6:} The pair $(T, \{B_i\}_{i \in I})$ forms a tree decomposition of $G$ of width $3k+2$.

\begin{proof}[Proof of Claim 6]
According to Definition \ref{de:treedec} we have to check that the following three properties are satisfied: 
\begin{itemize}
\item For every $v \in V$, there exists $i\in I$ such that $v \in B_i$;
\item For every $e = \{u,v\} \in E$ there is $i\in I$ such that $\{u, v\} \subseteq B_i$;
\item For every $v \in V$, the set $\{i \in I : v \in B_i\}$ forms a connected subgraph of $T$. 
\end{itemize}
The first two properties are directly verified as every vertex is given one bag that contains it in the main message. The second property is verified by condition {\bf \ref{cond:7}}. Finally, for the third condition, let us suppose that there exists a vertex $v\in V$ such that $I_v = \{i \in I: v \in B_i\}$ is not connected. Let $C_1$ and $C_2$ be two different components of $I_v$, and let $i_1$ and $i_2$ be, respectively, the nodes in $C_1$ and $C_2$ of minimum depth.   Observe that $F_{i_1} \neq F_{i_2}$ and by condition {\bf \ref{cond:2}}, $v$ must be contained in $F_{i_1} \cap F_{i_2}$, which contradicts {\bf Claim 2}. We deduce that for every $v \in V$, the set $\{i \in I : v \in B_i\}$ forms a connected subgraph of $T$. 
We conclude that $(T, \{B_i\}_{i \in I})$ forms a tree decomposition of $G$. Finally, the width of the decomposition is verified by condition {\bf 1}.
\end{proof}

We finish this section showing one more property of our verification algorithm, that is not required for the certification of the $3$-approximation of the treewidth, but will be useful in the next section. \\

\noindent {\bf Claim 7:} For every $v\in V$ and every $j \in \{1, \dots, d(v)\}$, the set $E_j(v)$ corresponds to the edges of graph induced by $B_j(v)$.

\begin{proof}[Proof of Claim 7]

We prove this claim by induction on $d(v)$. Suppose first that $d(v)=1$. Since $ F_1 = B_1$, we have that $F(u) = F(v)$ for every other vertex $u$ in $B_1$. By {\bf Claim 2} we obtain that $v$ and $u$ agree on the same set $E_1$. Then, by condition {\bf \ref{cond:4}} on all the vertices in $B_1$, we deduce that $E_1 = E[G_1]$. 
Now suppose that the claim is true for every vertex of depth smaller than $d>1$ and suppose that $d(v) = d$. By the induction hypothesis, for every $j \in \{1, \dots, d-1\}$ the set $E_{j}(v)$ corresponds to the set of edges of $G[B_{j}(v)]$. Then, it remains to prove that $E_d(v)$ corresponds to the set of edges of $G[B_d(v)]$.
Let $w_1, w_2$ be an arbitrary pair of vertices in $B(v)$, and call $d_1$ and $d_2$ the depth of $w_1$ and $w_2$, respectively. Without loss of generality assume that $d_1 \leq d_2$. By {\bf Claim 4} applied to all vertices in the path of nodes between $B_{d}(v)$ and $B_{d_2}(w_2)$, we have that $E_{d_2}(w_2) = E_{d_2}(v)$. By condition {\bf \ref{cond:5}}, we have that $w_1, w_2$ are adjacent if and only if $\{w_1, w_2\}$ belongs to $E_{d_2}(w_2)$. Suppose that $d_2 = d$. By {\bf Claim 2}, we know that all nodes in $F(v)$ have the same main messages, in particular, they agree in the set $E_{d}(v)$. Then $E_{d_2}(v) = E_d(v)$.  If $d_2 < d$, we have by condition {\bf \ref{cond:4}}, that  $w_1, w_2 \in E_{d_2}(v)$ if and only if $\{w_1, w_2\}$ belongs to $E_{d}(v)$.  In both cases we deduce that $\{w_1, w_2\} \in E$ if and only if $\{w_1, w_2\} \in E_d(v)$. 
\end{proof}

%% file: regular.tex

\section{Certifying regular properties}\label{se:regular}

In this section, we prove our main result, Theorem~\ref{theo:main-informal}. 

\begin{theorem}\label{theo:main-formal}
For every $k\geq 1$ and any regular graph property  $\cP(G)$, there exists a distributed certification protocol certifying that $\tw(G) \leq k$ and $\cP(G)$ is true, using certificates on $O(\log^2n)$ bits in $n$-node networks. 
\end{theorem}

For simplicity, we integrate the condition $\tw(G) \leq k$ to property $\cP$, by setting $\cP(G) = (\tw(G)\leq k) \wedge \cP(G)$. The new property is regular because property  $\tw(G)\leq k$ is regular (see, e.g.,~\cite{LMT19} for a discussion), and a conjunction of regular properties is regular by~\cite{BoPaTo92}. 
Basically, we enrich the protocol of Section~\ref{se:3approx} as follows. Either the protocol rejects because $\tw(G) >k$, or it constructs and certifies a tree decomposition at most $3k+2$. In the latter case, we also certify property $\cP$ using the tree decomposition of width $3k+2$ and the homomorphism classes $\cC$ of the property on $(3k+3)$-terminal graphs.


Fix the tree decomposition of width $3k+2$. As in the sketch of proof of Proposition~\ref{pr:gramtw}, for each node $i$ of the decomposition tree, $G_i$ denotes the $(3k+3)$-terminal graph corresponding to $G[V_i]$, with set of terminals $B_i$. Also, for each $w \in F_i$, let $Children(w)$ denote the set of children $j$ of $i$ such that $w$ is in charge of node $j$ (see Lemma~\ref{le:leaders} applied to $j$). In particular, the sets $Children(w)$ for $w \in F_i$ form a partition of the children nodes of $i$ in the decomposition tree. Denote by $G_i[w]$ the $(3k+3)$-terminal graph obtained from $G[B_i \cup \bigcup_{j \in Children(w)} V_j]$ by choosing $B_i$ as set of terminals. Note that if $Children(w)$ is empty, then $G_i[w]$ is simply the $3k+3$-terminal base graph $G^b_i$ corresponding to $G[B_i]$, as illustrated in Figure~\ref{fig:incharge}. 

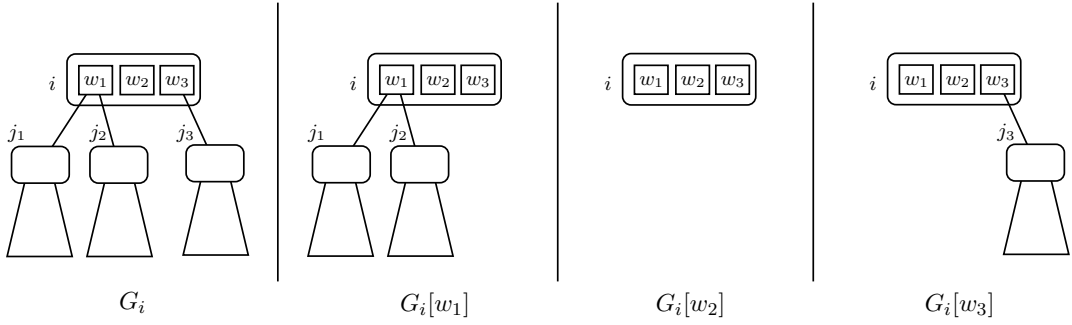
\begin{figure}[htbp]
\begin{center}
\scalebox{0.85}{
\input{Figures/ExGluParents.tex}
}
\caption{Graphs $G_i$ and $G_i[w]$.}
\label{fig:incharge}
\end{center}
\end{figure}


The \textbf{prover} appends two new informations to the previous main messages of each vertex $v \in F_i$: the homomorphism class of $G_i$ as well as the homomorphism class of $G_i[v]$. Moreover the homomorphism class of $G_i$ is also added to the auxiliary message $m_{aux}(w,i)$ for every vertex $w$ of the auxiliary tree $S(i)$.
Note that this only ads a constant size to the previous main messages, since property $\cP$ has a constant number of homomorphism classes. Auxiliary messages are increased by $O(\log n)$ bits, since each vertex $w$ is in $O(\log n)$ auxiliary trees $S(i)$ by Lemma~\ref{le:SteinerCongestion}. Nevertheless, the constants here depend on $k$ and on property $\cP$.

We now update the \textbf{verification round} to exploit these new messages and check the property. As before, we use the auxiliary tree $S(i)$ to ensure that $\ell_i$, and all vertices $v \in F_i$, have received from the prover the same isomorphism class for $G_i$.

It remains to check the consistency of the homomorphism classes for property $\cP$ in the respective subgraphs. 

\begin{itemize} 
\item \textit{Consistency of the homomorphism class of $G_i[v]$.} Firstly, each vertex $v \in F_i$ in charge of some nodes must certify the homomorphism class of $G_i[v]$, in the sense that it compares the message received from the prover with the homomorphism class that he constructs from the nodes $j \in Children(v)$. Vertex $v$ receives, for each $j \in Children[v]$, a message from $\ell_j$ with the homomorphism class of $\cP$ restricted to the $(3k+3)$-terminal graph $G_j$. Using Definition~\ref{de:reg}, it constructs the homomorphism class on $G^+_j$. Recall that $G^+_j = f(G_j,G^b_i)$, i.e., $G^+_j$ is obtained by glueing $G_j$ and the base graphs $G^b_i$ induced by $B_i$, the glueing being performed by identifying the terminals of $B_j \setminus B_i$ in $G_j$ to the corresponding vertices of $B_i$. Vertex $v$ knows both sets $B_i$ (which is in its initial message) and $B_j$ (received from $\ell_j$), so it has full knowledge of matrix $m(f)$ of the composition operation $f$. (There is a hidden technicality here. Node $\ell_j$ sends its main message to $v$ in the unique communication round, and this message contains all bags $\cB(\ell_j)$, in particular bag $B_j$. Node $v$ can retrieve this bag, since its order in the list $\cB(\ell_j)$, starting from the end of the list, is exactly the depth $i(v)$ of node $i$, plus one.) Then the homomorphism class of $h(G^+_j)$ is obtained as $\odot_f(h(G_j),h(G^b_i))$(see Figure~\ref{fig:incharge}, Proposition~\ref{pr:gramtw} and its sketch of proof). Again $v$ knows graph $G[B_i]$ hence it can compute its homomorphism class $h(G^b_i)$. It also knows $h(G^+_j)$ from $\ell_j$, altogether $v$ is able to compute the homomorphism class $h(G^+_j)$. Eventually, since $G[v]$ is obtained by glueing on $B_i$ all graphs $G^+_j, j \in Children(v)$, $v$ computes the homomorphism class of $G_i[v]$. If this class is not the same as the one received from the prover, vertex $v$ rejects. 

\item  \textit{Consistency of the homomorphism class of $G_i$.} Every vertex $v \in F_i$ checks the consistency between the message  received from the prover as class of $\cP$ on $G_i$, 
and the one it constructs from the glueing of all classes of $G_i[w]$ (that vertex $w$ has received from the prover), for all $w \in F_i$, on $B_i$. Indeed, $G_i$ is equal to the glueing, on $B_i$, of all graphs $G_i[w]$ with $w \in F_i$. Again, in case of inconsistency, vertex $v$ it rejects. 

\item \textit{Yes-instance.} Every vertex belonging to $F_r$ (the root node of the decomposition tree) accepts if the class of property $\cP$ on $G_r$ is an accepting one, otherwise it rejects. 
\end{itemize}

\paragraph{Soundness and Completeness.} For the completeness part, assume that our graph $G$ has treewidth at most $k$ and satisfies property $\cP$. By Lemma~\ref{lem:technique1-informal}, the prover can construct the messages for the 3-approximation of treewidth, such that the verifier passes all the tests certifying the tree decomposition. Moreover the tree decomposition is correct, and so are, for each node $i$ of the decomposition, the exit vertex $\ell_i$ of $i$ and the vertex of $F_{p(i)}$ in charge of node $i$. It remains to prove that vertices $v \in F_i$ accept. The proof is done bottom-up, by considering $i$ from the leaves to the root. If $i$ is a leaf of the decomposition tree, then $v$ is not in charge of any other node (i.e., $Children(v)$ is empty). In this case $G_i = G_i[v]$, and the homomorphism classes are all equal and correspond to the $(3k+3)$-terminal base graph $G[B_i]$, and all vertices $v \in F_i$ accept. Now if $i$ is not the root, every $v \in F_i$ is assigned a (possibly empty) set $Children(v)$ of children of $i$ in the decomposition tree. For each $j \in Children(v)$, vertex $v$ receives from $\ell_j$ the homomorphism class of $G_j$, so $v$ computes the class of $G^+_j$. By Proposition~\ref{pr:gramtw} and Definition~\ref{de:reg}, the homomorphism class of $G_i[v]$ is consistent withe the one obtained with the glueing of all $G^+_j$ on the set $B_i$ of terminals. Eventually, by glueing on $B_i$ all graphs $G_i[v]$, for all $v \in F_i$, we obtain $G_i$, and the homomorphism classes of $G_i$ and $G_i[v]$ are consistant, so $v$ accepts. At the root node $i=r$, each $v \in F_r$ also checks that the homomorphism class of $\cP$ is an accepting one (and it is), so $v$ accepts.

For the soundness, assume that all nodes accept. We must show by induction, for each node $i$ of the decomposition from leaves to the root, that the messages that each $v \in F_i$ received as homomorphism class of $\cP$ on graphs $G_i[v]$ and $G_i$ are correct. We rely again on the fact that the tree decomposition is correct, as well as the exit nodes and their neighbours in the parent node. When $i$ is a leaf node, each $v$ knows that its set $Children(v)$ is empty, since it has received no message from some exit node. Also $v$ knows the graph $G[B_i]$ (recall that all edges of $G[B_i]$ have been sent in the main messages). Therefore it checks that the homomorphism classes received from the prover for $G_i[v]$ and $G_i$ received are correct: they must be equal, and must correspond to the base graph $G[B_i]$. If $i$ is not a leaf node, we rely on the fact that, collecting the messages from the exit nodes $\ell_j$, vertex $v \in F_i$ correctly constructs $Children(v)$. For each $j \in Children(v)$, $v$ has received from $\ell_j$ the homomorphism class $c$ of $G_j$ (which is correct by induction hypothesis). Therefore $v$ correctly constructs the class of $G^+_j$ from $c$ and the class of the $(3k+3)$-terminal base graph $G[B_i]$. Then, by glueing all $G^+_j$, $j \in Children(v)$, $v$ it gets the class of $G_i[v]$. Since at this stage $v$ has not rejected, the class of $G_i[v]$ received from the prover is correct. Eventually, $v$ constructs the homomorphism class of $G_i$ by glueing the classes of all $G_i[w], w \in F_i$. Since $v$ knows $F_i$ and $B_i$, it correctly performs the glueing. By the fact that $v$ has not rejected up to now, we deduce that the homomorphism class of $G_i$ obtained from the prover is correct.

Since vertices $v$ of the root bag accept, il means that the homomorphism class of $\cP$ on the whole graph is an accepting one, so the property holds, which completes the proof of Theorem~\ref{theo:main-formal}.

%% file: Figures/ExGluParents.tex
\tikzset{every picture/.style={line width=0.75pt}} 

\begin{tikzpicture}[x=0.75pt,y=0.75pt,yscale=-1,xscale=1]

\draw   (1.79,170.24) -- (13.25,121.61) -- (28.54,121.61) -- (40,170.24) -- cycle ;
\draw   (48.68,170.24) -- (60.15,121.61) -- (75.43,121.61) -- (86.89,170.24) -- cycle ;
\draw   (105.13,169.37) -- (116.59,120.74) -- (131.88,120.74) -- (143.34,169.37) -- cycle ;
\draw  [fill={rgb, 255:red, 255; green, 255; blue, 255 }  ,fill opacity=1 ] (36.96,55.39) .. controls (36.96,52.63) and (39.2,50.39) .. (41.96,50.39) -- (110.12,50.39) .. controls (112.88,50.39) and (115.12,52.63) .. (115.12,55.39) -- (115.12,75.79) .. controls (115.12,78.55) and (112.88,80.79) .. (110.12,80.79) -- (41.96,80.79) .. controls (39.2,80.79) and (36.96,78.55) .. (36.96,75.79) -- cycle ;
\draw   (68.22,57.34) -- (87.91,57.34) -- (87.91,74.42) -- (68.22,74.42) -- cycle ;
\draw    (101.51,65.88) -- (123.8,115.09) ;
\draw  [fill={rgb, 255:red, 255; green, 255; blue, 255 }  ,fill opacity=1 ] (106.87,109.24) .. controls (106.87,106.48) and (109.11,104.24) .. (111.87,104.24) -- (135.74,104.24) .. controls (138.5,104.24) and (140.74,106.48) .. (140.74,109.24) -- (140.74,120.95) .. controls (140.74,123.71) and (138.5,125.95) .. (135.74,125.95) -- (111.87,125.95) .. controls (109.11,125.95) and (106.87,123.71) .. (106.87,120.95) -- cycle ;
\draw [line width=0.75]    (53.75,65.88) -- (21.33,115.96) ;
\draw  [fill={rgb, 255:red, 255; green, 255; blue, 255 }  ,fill opacity=1 ] (4.39,110.11) .. controls (4.39,107.34) and (6.63,105.11) .. (9.39,105.11) -- (33.26,105.11) .. controls (36.02,105.11) and (38.26,107.34) .. (38.26,110.11) -- (38.26,121.82) .. controls (38.26,124.58) and (36.02,126.82) .. (33.26,126.82) -- (9.39,126.82) .. controls (6.63,126.82) and (4.39,124.58) .. (4.39,121.82) -- cycle ;
\draw    (53.75,65.88) -- (67.36,115.96) ;
\draw  [fill={rgb, 255:red, 255; green, 255; blue, 255 }  ,fill opacity=1 ] (43.91,57.34) -- (63.59,57.34) -- (63.59,74.42) -- (43.91,74.42) -- cycle ;
\draw  [fill={rgb, 255:red, 255; green, 255; blue, 255 }  ,fill opacity=1 ] (50.42,110.11) .. controls (50.42,107.34) and (52.66,105.11) .. (55.42,105.11) -- (79.29,105.11) .. controls (82.05,105.11) and (84.29,107.34) .. (84.29,110.11) -- (84.29,121.82) .. controls (84.29,124.58) and (82.05,126.82) .. (79.29,126.82) -- (55.42,126.82) .. controls (52.66,126.82) and (50.42,124.58) .. (50.42,121.82) -- cycle ;
\draw  [fill={rgb, 255:red, 255; green, 255; blue, 255 }  ,fill opacity=1 ] (91.67,57.34) -- (111.36,57.34) -- (111.36,74.42) -- (91.67,74.42) -- cycle ;
\draw    (160.71,20) -- (160.71,185) ;
\draw   (178.45,170) -- (189.91,121.37) -- (205.19,121.37) -- (216.66,170) -- cycle ;
\draw   (225.34,170) -- (236.81,121.37) -- (252.09,121.37) -- (263.55,170) -- cycle ;
\draw  [fill={rgb, 255:red, 255; green, 255; blue, 255 }  ,fill opacity=1 ] (213.62,55.16) .. controls (213.62,52.4) and (215.86,50.16) .. (218.62,50.16) -- (286.78,50.16) .. controls (289.54,50.16) and (291.78,52.4) .. (291.78,55.16) -- (291.78,75.55) .. controls (291.78,78.31) and (289.54,80.55) .. (286.78,80.55) -- (218.62,80.55) .. controls (215.86,80.55) and (213.62,78.31) .. (213.62,75.55) -- cycle ;
\draw   (244.88,57.11) -- (264.57,57.11) -- (264.57,74.18) -- (244.88,74.18) -- cycle ;
\draw [line width=0.75]    (230.41,65.64) -- (197.99,115.72) ;
\draw  [fill={rgb, 255:red, 255; green, 255; blue, 255 }  ,fill opacity=1 ] (181.05,109.87) .. controls (181.05,107.11) and (183.29,104.87) .. (186.05,104.87) -- (209.92,104.87) .. controls (212.68,104.87) and (214.92,107.11) .. (214.92,109.87) -- (214.92,121.58) .. controls (214.92,124.34) and (212.68,126.58) .. (209.92,126.58) -- (186.05,126.58) .. controls (183.29,126.58) and (181.05,124.34) .. (181.05,121.58) -- cycle ;
\draw    (230.41,65.64) -- (244.01,115.72) ;
\draw  [fill={rgb, 255:red, 255; green, 255; blue, 255 }  ,fill opacity=1 ] (220.57,57.11) -- (240.25,57.11) -- (240.25,74.18) -- (220.57,74.18) -- cycle ;
\draw  [fill={rgb, 255:red, 255; green, 255; blue, 255 }  ,fill opacity=1 ] (227.08,109.87) .. controls (227.08,107.11) and (229.32,104.87) .. (232.08,104.87) -- (255.95,104.87) .. controls (258.71,104.87) and (260.95,107.11) .. (260.95,109.87) -- (260.95,121.58) .. controls (260.95,124.34) and (258.71,126.58) .. (255.95,126.58) -- (232.08,126.58) .. controls (229.32,126.58) and (227.08,124.34) .. (227.08,121.58) -- cycle ;
\draw  [fill={rgb, 255:red, 255; green, 255; blue, 255 }  ,fill opacity=1 ] (268.33,57.11) -- (288.01,57.11) -- (288.01,74.18) -- (268.33,74.18) -- cycle ;
\draw  [fill={rgb, 255:red, 255; green, 255; blue, 255 }  ,fill opacity=1 ] (363.08,55.16) .. controls (363.08,52.4) and (365.32,50.16) .. (368.08,50.16) -- (436.24,50.16) .. controls (439,50.16) and (441.24,52.4) .. (441.24,55.16) -- (441.24,75.55) .. controls (441.24,78.31) and (439,80.55) .. (436.24,80.55) -- (368.08,80.55) .. controls (365.32,80.55) and (363.08,78.31) .. (363.08,75.55) -- cycle ;
\draw   (394.34,57.11) -- (414.03,57.11) -- (414.03,74.18) -- (394.34,74.18) -- cycle ;
\draw  [fill={rgb, 255:red, 255; green, 255; blue, 255 }  ,fill opacity=1 ] (370.03,57.11) -- (389.71,57.11) -- (389.71,74.18) -- (370.03,74.18) -- cycle ;
\draw  [fill={rgb, 255:red, 255; green, 255; blue, 255 }  ,fill opacity=1 ] (417.79,57.11) -- (437.47,57.11) -- (437.47,74.18) -- (417.79,74.18) -- cycle ;
\draw   (586.79,168.97) -- (598.25,120.34) -- (613.54,120.34) -- (625,168.97) -- cycle ;
\draw  [fill={rgb, 255:red, 255; green, 255; blue, 255 }  ,fill opacity=1 ] (518.62,55) .. controls (518.62,52.24) and (520.86,50) .. (523.62,50) -- (591.78,50) .. controls (594.54,50) and (596.78,52.24) .. (596.78,55) -- (596.78,75.39) .. controls (596.78,78.16) and (594.54,80.39) .. (591.78,80.39) -- (523.62,80.39) .. controls (520.86,80.39) and (518.62,78.16) .. (518.62,75.39) -- cycle ;
\draw   (549.88,56.95) -- (569.57,56.95) -- (569.57,74.03) -- (549.88,74.03) -- cycle ;
\draw    (583.17,65.49) -- (605.46,114.7) ;
\draw  [fill={rgb, 255:red, 255; green, 255; blue, 255 }  ,fill opacity=1 ] (588.53,108.84) .. controls (588.53,106.08) and (590.76,103.84) .. (593.53,103.84) -- (617.39,103.84) .. controls (620.16,103.84) and (622.39,106.08) .. (622.39,108.84) -- (622.39,120.55) .. controls (622.39,123.31) and (620.16,125.55) .. (617.39,125.55) -- (593.53,125.55) .. controls (590.76,125.55) and (588.53,123.31) .. (588.53,120.55) -- cycle ;
\draw  [fill={rgb, 255:red, 255; green, 255; blue, 255 }  ,fill opacity=1 ] (525.57,56.95) -- (545.25,56.95) -- (545.25,74.03) -- (525.57,74.03) -- cycle ;
\draw  [fill={rgb, 255:red, 255; green, 255; blue, 255 }  ,fill opacity=1 ] (573.33,56.95) -- (593.01,56.95) -- (593.01,74.03) -- (573.33,74.03) -- cycle ;
\draw    (325,20) -- (325,185) ;
\draw    (475,20) -- (475,185) ;

\draw (45.39,62) node [anchor=north west][inner sep=0.75pt]  [font=\footnotesize]  {$w_{1}$};
\draw (69.71,62) node [anchor=north west][inner sep=0.75pt]  [font=\footnotesize]  {$w_{2}$};
\draw (93.16,62) node [anchor=north west][inner sep=0.75pt]  [font=\footnotesize]  {$w_{3}$};
\draw (24.88,62) node [anchor=north west][inner sep=0.75pt]  [font=\footnotesize]  {$i$};
\draw (48.76,90.76) node [anchor=north west][inner sep=0.75pt]  [font=\footnotesize]  {$j_{2}$};
\draw (1,90.76) node [anchor=north west][inner sep=0.75pt]  [font=\footnotesize]  {$j_{1}$};
\draw (100,90.76) node [anchor=north west][inner sep=0.75pt]  [font=\footnotesize]  {$j_{3}$};
\draw (222.05,62) node [anchor=north west][inner sep=0.75pt]  [font=\footnotesize]  {$w_{1}$};
\draw (246.37,62) node [anchor=north west][inner sep=0.75pt]  [font=\footnotesize]  {$w_{2}$};
\draw (269.82,62) node [anchor=north west][inner sep=0.75pt]  [font=\footnotesize]  {$w_{3}$};
\draw (201.54,62) node [anchor=north west][inner sep=0.75pt]  [font=\footnotesize]  {$i$};
\draw (225,90.52) node [anchor=north west][inner sep=0.75pt]  [font=\footnotesize]  {$j_{2}$};
\draw (177,90.52) node [anchor=north west][inner sep=0.75pt]  [font=\footnotesize]  {$j_{1}$};
\draw (371.51,62) node [anchor=north west][inner sep=0.75pt]  [font=\footnotesize]  {$w_{1}$};
\draw (395.83,62) node [anchor=north west][inner sep=0.75pt]  [font=\footnotesize]  {$w_{2}$};
\draw (419.28,62) node [anchor=north west][inner sep=0.75pt]  [font=\footnotesize]  {$w_{3}$};
\draw (351,62) node [anchor=north west][inner sep=0.75pt]  [font=\footnotesize]  {$i$};
\draw (527.05,62) node [anchor=north west][inner sep=0.75pt]  [font=\footnotesize]  {$w_{1}$};
\draw (551.37,62) node [anchor=north west][inner sep=0.75pt]  [font=\footnotesize]  {$w_{2}$};
\draw (574.82,62) node [anchor=north west][inner sep=0.75pt]  [font=\footnotesize]  {$w_{3}$};
\draw (506.54,62) node [anchor=north west][inner sep=0.75pt]  [font=\footnotesize]  {$i$};
\draw (581.66,90.36) node [anchor=north west][inner sep=0.75pt]  [font=\footnotesize]  {$j_{3}$};
\draw (66,190) node [anchor=north west][inner sep=0.75pt]    {$G_{i}$};
\draw (231,190) node [anchor=north west][inner sep=0.75pt]    {$G_{i}[ w_{1}]$};
\draw (381,190) node [anchor=north west][inner sep=0.75pt]    {$G_{i}[ w_{2}]$};
\draw (539,190) node [anchor=north west][inner sep=0.75pt]    {$G_{i}[ w_{3}]$};

\end{tikzpicture}

%% file: conclusion.tex

\section{Conclusion}\label{sec:conclusion}

To sum up, we proved that for every $k\geq 1$ and every MSO property on graphs, there exists a distributed protocol certifying that the input graph is of treewidth at most $k$ and satisfies the required property, using certificates on $O(\log^2n)$ bits. The result extends to optimisation problems, where we certify that a given vertex subset is of optimal weight (e.g., of maximum or of minimum size) for some MSO property, and the treewidth of the input graph is at most $k$.

The first natural question is whether we can reduce the size of certificate to $O(\log n)$ instead of $O(\log^2 n)$. We believe that such an improvement requires considerably different techniques, even for certifying that the treewidth of the input graph is at most $k$.

Another further research direction concerns certification versions for other algorithmic ``meta-theorems''. For example, given a graph property expressible by a first-order boolean formula, is there a distributed protocol certifying that the input graph is planar and satisfies the property, using certificates of logarithmic size?

\paragraph{Acknowledgment.} The authors are thankful to Eric Remila for fruitful discussions on certification schemes related to the one considered in this paper.

%% file: appendixOptim.tex

\section{More preliminaries: MSO and regular properties for optimization}\label{app:propt}


Let us enrich our framework to properties on graphs and vertex subsets, i.e., properties $\cP(G,X)$ assigning to each graph $G$ and each vertex subset $X$ of $G$ a boolean value. Properties like "$X$ is an independent set of $G$" or "$X$ is a dominating set of $G$" are expressible in (Counting) Monadic Second Order Logic, and they are still regular as we shall see below. More importantly, in the sequential setting this allows to solve efficiently optimisation problems on graphs of bounded treewidth, namely to compute a vertex subset $X$ of maximum (or minimum) size such that  $\cP(G,X)$ holds.

Composition operations on $w$-terminal recursive graphs naturally extend to pairs $(G,X)$, where $G$ is a $w$-terminal recursive graph and $X$ is a vertex subset. Let $f$ be a composition operation of arity 1, and $G = f(G_1)$. Then for every vertex subset $X_1$ of $G_1$, we take $f(G_1,X_1) = (G, X_1)$. Consider composition operation of arity 2 such that $G = f(G_1, G_2)$. When we perform this composition on pairs $(G_1,X_1)$, $(G_2, X_2)$, the result is the pair $(G,X)$, where $X$ is obtained by the the glueing of $X_1$ and $X_2$. Therefore the intersections of sets $X_1$ and $X_2$ with the terminals of $G_1$ and respectively $G_2$ must be coherent with the gluing, in the sense that if two terminals $x_1$ of $G_1$ and $x_2$ of $G_2$ are identified in $G$, then we either have $x_1 \in X_1$ and $x_2 \in X_2$, or we have $x_1 \notin X_1$ and $x_2 \notin X_2$ (see~\cite{BoPaTo92,FoToVi15} for more details). To be complete, we restate the notion of regularity to properties $\cP(G,X)$ -- the only difference being that the property and the homomorphism classes now depend on both parameters, the graph and the vertex subset.

\begin{definition}[regular property on graphs and vertex sets]\label{de:regX}
Graph property $\cP(G,X)$ is called \emph{regular} if, for any value $w$, we can associate a finite set $\cC$ of \emph{homomorphism classes} and a \emph{homomorphism function} $h$, assigning to each $w$-terminal recursive graph $G$ and to each vertex subset $X$ a class $h(G,X) \in \cC$ such that:
\begin{enumerate}
\item If $h(G_1,X_1) = h(G_2,X_2)$ then $\cP(G_1,X_1) = \cP(G_2,X_2)$.
\item For each composition operation $f$ of arity 2 there exists a function $\odot_f: \cC \times \cC \rightarrow \cC$ such that, for any two pairs $(G_1,X_1)$ and $(G_2,X_2)$,
$$h(f((G_1,X_1),(G_2,X_2)) = \odot_f(h(G_1,X_1),h(G_2,X_2))$$
and for each composition operation $f$ of arity 1 there is a function $\odot_f: \cC \rightarrow \cC$ such that, for any pair $(G,X)$,
$$h(f(G,X)) = \odot_f(h(G,X)).$$
\end{enumerate}
\end{definition}
The first condition separates the classes into \emph{accepting} ones (i.e., classes $c \in \cC$ such that $h(G,X)=c$ implies that $\cP(G,X)$ is true) and \emph{rejecting} ones (s.t. $h(G,X)=c$ implies that $\cP(G,X)$ is false).

We also have:

\begin{proposition}[\cite{BoPaTo92,Courcelle90}]\label{pr:regX}
Any property $\cP(G,X)$ expressible by a $MSO$ formula is regular.

Moreover, given the $MSO$ formula $\varphi$ and parameter $w$, one can explicitely compute the set of classes, the homomorphism function for all $w$-terminal base graphs as well as the homomorphism functions $\odot_f$  over all possible composition operations $f$. 
\end{proposition}

E.g., for the property "$X$ is an independent set of $G$", we can choose as homomorphism class $h(G(V,W,E),X)$ formed by a boolean indicating whether $\cP(G,X)$ is true, and the intersection of $X$ with the set of terminals. 
 
We may assume w.l.o.g. that the homomorphism class  $c=h(G,X)$, for $G=(V,W,E)$ always encodes the intersection of $X$ with the set of terminals. This is not explicitly required by the definition of regular properties, but it can be done since it only costs $w$ bits to encode the subset of the terminals contained in $X$. Therefore we assume there is a function $term(c,W)$ that, given a homomorphism class $c$ and an ordered set of terminals $W$ returns the unique possible set $X \cap W$, over all pairs $(G=(V,W,E),X)$ mapped to $c$. Thanks to this function, when we glue two terminal recursive graphs with their corresponding vertex subsets, we will be able to check that the glueing is coherent. Moreover, we can perform optimisation tasks as follows.

Assume that we deal with weighted graphs, i.e., we have a function $weigth$ associating to every vertex an integer weight in the interval $[-MAXW,+MAXW]$. Given a $w$-terrminal recursive graph $G$ and a regular property $\cP(G,X)$, we aim to compute the maximum weight vertex subset $X$ satisfying $\cP(G,X)$. For this purpose, for any homomorphism class $c \in \cC$ of property $\cP$, let $$MaxWeight(G,c) = \max\{weight(X) \mid X \subseteq V(G) \hbox{ s.t. } h(G,X) = c\}.$$ 
For convenience, we set $MaxWeight(G,c)$ to $-\infty$ if no such set exists, and more generally we consider that the maximum value of an empty set is $-\infty$. Then we have:

\begin{lemma}\label{le:optim}
For any $w$-terminal recursive graph $G=(V,W,E)$ and any homomorphism class $c$ of property $\cP$,
\begin{enumerate}
\item If $G$ is a $w$-terminal base graph,
$$ MaxWeight(G,c)  = weight(term(c,W)).$$
\item If $G = f(G_1)$ for some composition operation $f$ of arity 1, then
$$ MaxWeight(G,c)  = \max_{c_1  \hbox{ s.t. } c = \odot_f(c_1)} MaxWeight(G_1,c_1) .$$
\item If $G = f(G_1, G_2)$ for some composition operation $f$ of arity 2, then
\begin{align*}
MaxWeight(G,c)  = \max_{c_1, c_2  \hbox{ s.t. } c = \odot_f(c_1, c_2)} & MaxWeight(G_1,c_1) + MaxWeight(G_1,c_1) - \\
	&weight(term(c_1,W_1) \cap (term(c_2,W_2)),
\end{align*}
where $W_j$ denotes the set of terminals of graph $G_j$, for $j \in \{1,2\}$.
\end{enumerate}
\end{lemma}
\begin{proof}
The first two items are simple consequences of the definitions, let us focus on the third item. 

Firstly, let us prove that $MaxWeight(G,c)$ is at least equal to the right-hand side of the expression. Let $c_1$, $c_2$ be the homomorphism classes realising the maximum, and for each $j \in \{1,2\}$ let $X_j$ be the vertex subset of $G_j$ such that $h(G_j,X_j) = c_j$ and $MaxWeight(G_j,c_j) = weight(X_j)$. Observe that, by taking the vertex subset $X$ of $G$ obtained from $X_1$ and $X_2$ by glueing the corresponding terminal vertices according to composition operation $f$, $weight(X) = weight(X_1) + weight(X_2) - weight(term(c_1,W_1) \cap (term(c_2,W_2))$, since the negative term avoids overcounting the vertices of $X$ appearing as terminals in both $X_1$ and $X_2$. By construction $(G,X) = f((G_1,X_1),(G_2,X_2)$ so $MaxWeight(G,c)$ is at least $weight(X)$.

Conversely, let $X$ be a maximum weight vertex subset of $G$ such that $h(G,X) = c$. For $j\in \{1, 2\}$, let $X_j$ be the intersection of $X$ with the vertex set of $G_j$, and $c_j = h(G_j,X_j)$. By construction, $c = \odot_f(c_1,c_2)$ and $weight(X) = weight(X_1) + weight(X_2) - weight(term(c_1,W_1) \cap (term(c_2,W_2))$. We claim that $weight(X_j) = MaxWeight(G_j,c_j)$, for both values $j$. Assume by contradiction there is a set, say $X'_1$, of larger weight than $X_1$ and such that $h(G_1,X_1) = c_1$. Note that $X'_1$ and $X_1$ may only differ on non-terminal vertices of $G_1$, otherwise they would not correspond to the same homomorphism class. Then set $X'$ obtained by glueing $X'_1$ and $X_2$ is of larger weight than $X$, moreover $h(G,X') = c$, contradicting the maximality of $X$. We conclude that the right-hand side of the expression is at least equal to   $MaxWeight(G,c)$ and the conclusion follows.

\end{proof}

\section{Certifying optimal sets for regular properties}\label{app:plsopt}

We can now extend out certification protocol to optimisation problems on weighted graphs, with polynomial weights.
\begin{theorem}\label{theo:main-formal-opt}
For every $k\geq 1$ and any regular graph property $\cP(G,X)$, there exists a distributed certification protocol certifying that $\tw(G) \leq k$ and $X$ is the maximum weight vertex set such that $\cP(G,X)$ is true, on graphs with polynomial weights, using certificates on $O(\log^2n)$ bits in $n$-node networks. 
\end{theorem}

If instead of polynomial weights we use weights in the interval $[-MAXW,+MAXW]$, the protocol requires $O(\log n(\log n + \log MAXW))$ bits. 

We only describe the differences of the new protocol with respect to the protocol of Section~\ref{se:regular}. As for the protocol of Theorem~\ref{theo:main-formal}, we assume that the condition $\tw(G) \leq k$ is integrated to property $\cP$. Here the input is also formed of vertex set $X$. On one hand we certify $\cP(G,X)$ (this part of the protocol being almost identical to the one of Theorem~\ref{theo:main-formal}), and in the meantime we certify, for each homomorphism class $c$ and at each node $i$ of the decomposition, the weight of an optimal partial solution $(G_i,Y)$ for graph $G_i$, of homomorphism class $c$. Then we simply compare at the root node the weight of $X$ with the weight of an optimal solution.


Let us detail how we deal with set $X$. 

The first issue is that, for each node $i$ of the decomposition tree and each vertex $v \in F_i$, vertex $v$ must know the set $B_i \cap X$. For this purpose, 
The prover adds to the main messages of vertex $v$, a sequence of sets $\mathcal{X}(v) = (X_d(v), \dots X_1(v))$ where $X_j(v)$ represents the intersection of the solution $X$ with bag $B_j(v)$. 

The verification is very similar to the one of the edge sets of $G[B_i]$. In the verification round, $v$ verifies for each $j_1, j_2 \in \{1, \dots, d\}$ and for each $u \in B_{j_1}(v) \cap B_{j_2}(v)$, that $u \in X_{j_1}(v) \iff u \in X_{j_2}(v)$. By {\bf Claim 2}, all vertices in $F(v)$ receive the same main messages, then all nodes in $F(v)$ agree in the part of the solution $X$ that intersect the bags in the nodes from $i = B_d(v)$ up to the root. By {\bf Claim 4} the vertices globally agree on the set $X$. Also, each vertex $v \in X$ verifies that it belongs to $X_d(v)$, ensuring that the set $X$ claimed by the prover is consistent with the input. 

A second issue to deal with is the overall weight of set $X$. Here we use a completely different but standard technique to collect the weight of $X$ in a vertex $v_r$ belonging to the root bag, using $O(\log n)$ supplementary bits per vertex, see~\cite{KormanKP10}. We encode a spanning tree of the whole graph rooted in $v_r$, by giving to each vertex its distance to the root and the identifier to the parent vertex. Moreover, each vertex $v$ receives the total weight $weightX(v)$ of the vertices of $X$ contained in the subtree rooted at $v$. The situation is very similar to the tree certificates $TreeCert$ that we have used in Section~\ref{se:3approx} for the auxiliary messages, where we encoded a subtree $S(i)$ rooted in a given vertex $\ell_i$ and spanning the vertex subset $F_i$. The verification follows exactly the same principles for certifying the tree structure, moreover each vertex $v$ checks that  $weightX(v)$ corresponds to the sum of weights $weightX(w)$ for its children $w$, plus the weight of $v$ if the latter belongs to $X$.

A third issue is that, for each node $i$ of the decomposition tree, the prover sends to each $v \in F_i$ and the exit vertex $\ell_i$ the homomorphism class $h(G_i,X \cap V_i)$ (instead of the class of $G_i$).  
Also, $v$ receives the homomorphism class $h(G_i[v], X \cap V(G_i[v])$ and the weight of $X \cap V(G_i[v])$. This part is a simple update of the the protocol of 
Theorem~\ref{theo:main-formal}, adapted to properties on graphs and edge subsets.

The main novelty is that each $v \in F_i$ and $\ell_i$ receive from the prover, for each homomorphism class of property $\cP$, value $MaxWeight(i,c)$ corresponding to the maximum weight of a partial solution $(G_i, Y)$ of homomorphism class $c$ on graph $G_i$, and $v$ also receives value $MaxWeight(i,c; v)$, the maximum weight of $Y \subseteq V(G_i[v])$ such that $h(G_i[v],Y) = c$.

Let describe the verification performed by each vertex. We already ensured that vertices of a same set $F_i$, for each node $i$ of the decomposition tree, posses the correct set $X \cap B_i$. Checking property $\cP(G,X)$ is has no significant difference compared to Theorem~\ref{theo:main-formal}, we simply use the homomorphism functions of Definition~\ref{de:regX} instead of Definition~\ref{de:reg}. The construction is similar, we simply use the fact that each vertex $v \in F_i$ knows $X \cap B_i$, allowing it to compute the homomorphism class for base graphs $G^b_i$.

A supplementary effort is required to compute the weight of an optimal solution, then to compare it to the weight of $X$. For this purpose, at each node $i$ of the decomposition tree, the verifier performs the following operations on each $v \in F_i$.
\begin{itemize} 
\item Firstly, vertex $v$ checks that values $MaxWeight(i,c; v)$ received from the prover for each homomorphism class $c \in C$, claimed to be equal to $MaxWeight(G_i[v],c)$ are indeed consistent with the information it receives from nodes $j \in Children[v]$, the graph $G[B_i]$ and $X \cap B_i$.

For this purpose, $v$ computes $MaxWeight(G^b,c)$, for each homomorphism class $c$, using Lemma~\ref{le:optim} applied to the $3k+3$-terminal base graph $G^b$ and set $X \cap B_i$. Recall that $v$ has $G^b = (B_i, B_i, E(G[B_i]))$ and $X \cap B_i$ in its own message. Then, for each $j \in Children(v)$, it retrieves all values $MaxWeight(G_j,c)$ from the exit vertex $\ell_i$. Using again Lemma~\ref{le:optim} for graph $G^+_j = f(G_j,G^b_i)$, it computes all values $MaxWeight(G^+_j,c)$, from $MaxWeight(G_j,c_1)$, $MaxWeight(G_j,c_2)$ and $weight(term(c_1,B_j)) \cap term(c_2,B_i))$, over all classes $c_1, c_2$ with $c = \odot_f(c_1,c_2)$.

Then $v$ must deduce $MaxWeight(G_i[v],c)$ based on the fact that $G_i[v]$ is obtained by consecutively glueing $G^+_{j_1}$, $G^+_{j_2},\dots,$ $G^+_{j_p}$, where $Children(v) = \{j_1,\dots, j_p\}$ (e.g., we can order the nodes $j$ of $Children(v)$ by increasing size of the identifier of $\ell_j$). The glueing (composition operation) $f$ is the same at each step, performed on the same set of terminals $B_i$. Let $H^r$ denote the result of the glueing of $G^+_{j_1} \dots,G^+_{j_r}$, for each $r, 1 \leq r \leq p$. In particular $H^1 = G^+{j_1}$, $H^r = G_i[v]$ and $H^r = f(H^{r-1},G_{j_r})$ for each $r, 2 \leq r \leq p$. Therefore, for each $r$ from $2$ to $p$, vertex $v$ computes all values $MaxWeight(H^r,c)$ from values $MaxWeight(H^{r-1},c_1)$ and $MaxWeight(G_{j_r},c_2)$ using the equation of Lemma~\ref{le:optim} on operation $f$. Eventually $v$ has all values $MaxWeight(G_i[v],c)$ for all homomorphism classes $c$. If one of these values differs from the message $MaxWeight(i,c; v)$ received from the prover, then $v$ rejects.

\item Secondly, vertex $v$ checks that values $MaxWeight(i,c)$ correspond, for each homomorphism class $c$, to the value $MaxWeight(G_i,c)$ obtained by expressing $G_i$ as the consecutive glueing of all $G_i[w]$, for all $w \in F_i$, on the set of terminals $B_i$. Value $MaxWeight(G_i,c)$ is obtained by iteratively performing the $|F_i|-1$ glueings of $G_i[w]$ and using Lemma~\ref{le:optim} and values $Max(i,c;w)$. As above, at iteration $r$, $2 \leq r \leq |F_i|$, we glue the first $r$ graphs of the form $G[w]$, where vertices $w$ are ordered by increasing identifiers. Again, in case of inconsistency between $MaxWeight(i,c)$ and $MaxWeight(G_i,c)$ for some homomorphism class $c$, vertex $v$ rejects.

\item  At the root node $r$, recall that we must have a vertex $v_r \in F_r$ that knows the weight of $X$ -- it corresponds simply to $weightX(v_r)$. The node $v_r$ firstly checks that it belongs indeed to the root of the decomposition tree by testing its depth, i.e., checking that $d(v_r)=1$. Then $v_r$ computes the maximum weight $MaxWeight(G_r,c)$ as the maximum of $MaxWeight(r,c)$ over all accepting classes $c$. If one of those is larger than the $weightX(v_r)$, vertex $v_r$ rejects.
\end{itemize}

\paragraph{Soundness and completeness.} We already know that the protocol correctly encodes the tree decomposition, the homomorphism classes of $\cP(G,X)$ on partial graphs $G_i$ and $G_i(v)$, and that the weight of set $X$ is encoded in $weightX(v_r)$ for some vertex $v_r$ belonging to the root bag. It remains to deal with quantities  $MaxWeight(i,c)$ and $MaxWeight(i,c;v)$.

For the completeness part, the prover simply needs to correctly compute the intersection of $X$ with the bags, and values $MaxWeight(i,c)$ and $MaxWeight(i,c;v)$ for each homomorphism class $c$, each node $i$ of the decomposition tree and each vertex $v \in F_i$. The proof that vertex $v$ accepts when certifying messages $MaxWeight(i,c)$ and $MaxWeight(i,c;v)$ assigned to it follows the same steps as the completeness part for the decision problem, certifying that homomorphism classes of $G_i$ and $G_i[v]$ are correct. We need to use Lemma~\ref{le:optim}, allowing to obtain the optimal weight of a homomorphism class after glueing, instead of simply using Definition~\ref{de:reg}. Therefore, we prove by bottom-up induction on nodes $i$ that all vertices $v \in F_i$ accept, if $i$ is not the root. When $i$ is the root $r$, vertex $v_r \in F_r$ also check that the homomorphism class of $\cP(G_r,X)$ is an accepting one, and that the weight of $X$ corresponds to the maximum weight of an accepting class, and both conditions hold for a yes-instance. 

For the soundness condition, assume that all vertices accept. We prove as before, by bottom-up induction (from leaves to the root) on nodes $i$, that homomorphism classes as well as quantities $MaxWeight(i,c)$ and $MaxWeight(i,c;v)$ are correct, in the sense that they correspond to graphs $G_i$ and $G_i[v]$. Again, for values $MaxWeight(i,c)$ and $MaxWeight(i,c;v)$, the glueing is performed using Lemma~\ref{le:optim}.

At the root, since vertex $v_r \in F_r$ accepts, it means that $\cP(G_r,X)$ is true and moreover the weight of $X$ (which is equal to $weightX(v_r)$) is the maximum possible weight over all vertex subsets $Y$ such that $\cP(G,Y)$ accepts (by the last item of the verification protocol). Therefore $X$ is the optimal set for property $\cP$.